\documentclass{article}
\pdfoutput=1
\usepackage[utf8]{inputenc}
\usepackage{tikz}
\usepackage{authblk}
\usepackage{fullpage}
\usepackage{paralist}
\usepackage{amsmath,amsthm,amssymb}
\usepackage[scaled=.96,osf]{XCharter}%
\usepackage[charter,expert]{mathdesign}
\pdfoutput=1
\usetikzlibrary{positioning,backgrounds,patterns,calc}

\usepackage{mathtools}
\usepackage{enumerate}
\usepackage{booktabs}
\usepackage{tabularx}
\usepackage{multirow}
\usepackage{hyphenat}
\usepackage{tikz}
\usepackage{pgf}
\usepackage{subcaption}
\usepackage{pgfplots}
\usepackage[textsize=tiny]{todonotes}
\usepackage{hyperref}
\usepackage{doi}
\usepackage[sort&compress,noabbrev, capitalize, nameinlink]{cleveref}
\usepackage{verbatim}
\pgfplotsset{compat=1.14}

\hypersetup{colorlinks,allcolors=blue}

\usepackage[numbers,sort&compress]{natbib}

\crefname{problem}{Problem}{Problems}
\crefname{transformation}{Transformation}{Transformations}
\crefname{rrule}{Reduction Rule}{Reduction Rules}
\crefname{part}{Part}{Parts}
\crefname{proposition}{Proposition}{Proposition}
\Crefname{proposition}{Prop.}{Props.}

\newcommand{\plb}{vertex lower bound}
\newcommand{\Plb}{Vertex lower bound}
\newcommand{\pG}{G_{\ell}^{\bullet}}
\newcommand{\pE}{E_{\ell}^{\bullet}}
\newcommand{\ppE}{E_{\ell}^{\circ}}
\newcommand{\pw}{w_{\ell}^{\bullet}}
\newcommand{\pI}{I^{\bullet}}
\newcommand{\N}{\mathbb{N}}

\newcommand{\col}{\text{col}}

\newcommand{\miPoSyCo}{\textsc{MinPSC}}

\newcommand{\MCCS}{\textsc{MinPICCS}}
\newcommand{\miPoSyCoAl}{\textsc{MinPSC-ALB}}
\newcommand{\dmiPoSyCoAl}{\(d\)-\textsc{PSC-ALB}}
\newcommand{\kmiPoSyCo}{$k$-\textsc{PSC}}
\newcommand{\mpsc}{\textsc{\miPoSyCo}}

\newcommand{\SC}{\textsc{Set Cover}}
\newcommand{\MSC}{\textsc{Minimum Set Cover}}
\newcommand{\nc}{z}
\DeclareMathOperator{\opt}{Opt}
\DeclareMathOperator{\mar}{Opt_{mar}}

\theoremstyle{definition}
\newtheorem{theorem}{Theorem}[section]
\newtheorem{proposition}[theorem]{Proposition}
\newtheorem{assumption}[theorem]{Assumption}
\newtheorem{lemma}[theorem]{Lemma}

\newtheorem{observation}{Observation}
\newtheorem{problem}[theorem]{Problem}
\newtheorem{algorithm}{Algorithm}
\newtheorem{definition}[theorem]{Definition}

\newtheorem{transformation}{Transformation}
\newtheorem{remark}[theorem]{Remark}
\newtheorem{example}[theorem]{Example}

\newcommand{\optprob}[3]{
  \pagebreak[3]
  \begin{problem}[\textsc{#1}]\leavevmode

    \noindent
    \textit{Input:}      #2

    \noindent
    \textit{Goal:}
      #3
  \end{problem}
}
\newcommand{\alValue}{margin}

\title{Parameterized Algorithms for Power\hyp Efficiently Connecting Wireless Sensor Networks:
  Theory and Experiments\thanks{A preliminary version of this article
    appeared
    in the
    \emph{Proceedings of the 13th International Symposium on Algorithms and Experiments for Wireless Networks (ALGOSENSORS’17), Vienna, Austria} \citep{BBNN17}.
    This extended version is enhanced
    by an experimental evaluation
    and
    by results translated
    from the last author's Master's thesis \citep{Smi20}: it contains stronger lower-bound results and
    a corrected and accelerated algorithm.}}

\author[1]{Matthias Bentert}
\affil{Algorithmics and Computational Complexity, Faculty~IV, TU Berlin, Berlin, Germany,\newline
  \texttt{\{matthias.bentert, andre.nichterlein, rolf.niedermeier\}@tu-berlin.de}}
\author[2]{René van Bevern}
\affil{Department of Mechanics and Mathematics,
  Novosibirsk State University,
  Novosibirsk,
  Russian Federation,\newline
  \texttt{rvb@nsu.ru, p.smirnov@g.nsu.ru}}
\author[1]{André Nichterlein}
\author[1]{Rolf Niedermeier}
\author[2]{Pavel V.\ Smirnov}

\begin{document}
\maketitle

\begin{abstract}
  We study an NP-hard problem motivated by
  energy-efficiently maintaining
  the connectivity of
  a symmetric wireless communication network:
  Given an edge-weighted \(n\)-vertex graph,
  find a connected spanning subgraph of minimum cost,
  where the cost is determined by letting each vertex
  pay the most expensive edge incident to it in the subgraph. 
  On the negative side, we show that
  \(o(\log n)\)-approximating the difference~\(d\)
  between the optimal solution cost
  and a natural lower bound is NP-hard
  and that, under the Exponential Time Hypothesis,
  there are no exact algorithms 
  running in \(2^{o(n)}\)~time
  or in \(f(d)\cdot n^{O(1)}\)~time
  for any computable function~\(f\).
  Moreover,
  we show that
  the special case
  of connecting $c$~network components
  with minimum additional cost
  generally cannot be polynomial\hyp time reduced
  to instances of size $c^{O(1)}$
  unless the polynomial\hyp time hierarchy collapses.
  On the positive side,
  we provide an algorithm
  that reconnects \(O(\log n)\)~connected components
  with minimum additional cost in polynomial time.
  These algorithms are 
  motivated by application scenarios
  of monitoring areas or
  where an existing sensor network
  may fall apart into several connected components
  due to sensor faults.
  In experiments,
  the algorithm %
  outperforms CPLEX with
  known ILP formulations
  when $n$~is sufficiently large compared to~$c$.
\end{abstract}

\section{Introduction}
We consider a well-studied %
problem arising in the context of
power\hyp efficiently
maintaining the connectivity of 
symmetric wireless sensor communication networks.
It is formally defined as follows
(see \cref{fig:exampleMPSC} for an example).

\optprob{\textsc{Min-Power Symmetric Connectivity} (\miPoSyCo{})%
}%
{\label{prob:miposyco}A connected undirected graph~$G = (V,E)$ with \(n\)~vertices, \(m\)~edges, and edge weights~$w\colon E \to \N$.}%
{Find a connected spanning subgraph~$T=(V,F), F \subseteq E,$ of~$G$ that minimizes
  \begin{align*}
    \sum_{v \in V} \max_{\{u,v\}\in F}w(\{u,v\}).%
  \end{align*}
}

Herein,
a \emph{spanning} subgraph of a graph~$G=(V,E)$ is a subgraph
that contains all vertices~$V$.
We denote the minimum cost of a solution to an \miPoSyCo{} instance~\(I=(G,w)\) by~\(\opt(I)\).
Throughout this work,
\emph{weights} always refer to edges
and
\emph{cost} refers to vertices or subgraphs.
For proving computational hardness results,
we will also consider the \emph{decision version}
of \miPoSyCo{}, which we call \kmiPoSyCo{}.
Herein, given a natural number~$k$, the problem is to decide
whether a \miPoSyCo{} instance~\(I=(G,w)\)
satisfies \(\opt(I)\leq k\).

\begin{figure}
  \centering
  \begin{tikzpicture}
    \node[circle,draw, label=$5$] (v1) [] {$v_1$};
    \node[circle,draw, label=$6$] (v2) [right= of v1] {$v_2$};
    \node[circle,draw, label=$6$] (v3) [right= of v2,xshift=4cm] {$v_3$};
    \node[circle,draw, label=$5$] (v4) [right= of v3] {$v_4$};
    \node[circle,draw, label=below:$1$] (v5) [below= of v2] {$v_5$};
    \node[circle,draw, label=below:$3$] (v6) [below= of v3] {$v_6$};
    \path[draw, ultra thick] (v1)-- node [above]{$5$}(v2);
    \path[draw, ultra thick] (v2)-- node [above]{$6$}(v3);
    \path[draw, ultra thick] (v3)-- node [above]{$5$}(v4);
    \path[draw] (v2)-- node [left]{$4$}(v5);
    \path[draw, ultra thick] (v3)-- node [right]{$3$}(v6);
    \path[draw, ultra thick] (v5)-- node [below]{$1$}(v6);
  \end{tikzpicture}
  \caption{A graph with positive edge weights and an optimal solution (bold edges).
    Each vertex pays the most expensive edge incident to it in the solution (the numbers next to the vertices).
    The cost of the solution is
    the sum of the costs paid by the vertices.
    Note that the optimal
    solution has cost~$26$ while a minimum spanning tree
    (using edge~\(\{v_2,v_5\}\) instead of edge~$\{v_2,v_3\}$)
    has cost~$27$ (as a \miPoSyCo{} solution).}
  \label{fig:exampleMPSC}
\end{figure}

\mpsc{} falls into the category of survivable network design~\citep{Pan11}.
We refer to \citet{CHPRV02} and \citet{CMZ02} for a survey on different application scenarios and some results to related problems.
Notably, \cref{fig:exampleMPSC} reveals that computing 
a minimum-cost spanning tree typically does not yield an optimal 
solution for \miPoSyCo{} (\citet{EPS13} and \citet{GBHO19} provide further 
discussions concerning the relationship to minimum-cost spanning trees).
Indeed, \miPoSyCo{} is NP-hard \citep{KKKP00,CPS04}.
In this work, we provide a refined
computational complexity analysis
of \miPoSyCo{} in terms of parameterized complexity theory.
In this way, we complement previous findings mostly concerning %
polynomial-time approximability~\citep{KKKP00,ACM+06,CPS04,EPS13,GBHO19}, heuristics and 
integer linear programming~\citep{ACM+06,EMP17,MG05,PEM19}, and computational 
complexity analysis for special cases~\citep{CK07,CPS04,EPS13,HW16}.
We also complement recent positive results on provable effective
data reduction for \mpsc{} \citep{BBF+xx} by negative ones.

\begin{table}[t]
   \caption{Overview on our results, using the following terminology:
     \(n\)---number of vertices,
     \(m\)---number of edges,
     \(d\)---difference between the optimal solution cost and a lower bound (see \cref{prob:mpscal}),
     \(c\)---number of connected components of the subgraph consisting of obligatory edges (see \cref{def:oblig}).
     \miPoSyCoAl{} is the problem of computing the minimum value of~\(d\) (\cref{prob:mpscal}),
     \dmiPoSyCoAl{} is the corresponding decision problem.}
  \label{tab:results}
  \setlength{\tabcolsep}{2.8pt}
  \renewcommand{\arraystretch}{1.25}
	\begin{tabularx}{\textwidth}{crXl}
		\toprule
		& Problem & Result & Reference\\
		\midrule
		
		\multirow{3}{*}{\rotatebox[origin=c]{90}{Sec.~\ref{sec:hardness}}}
                & \miPoSyCoAl & NP-hard to approximate within a factor of~$o(\log n)$ &
                \cref{thm:hard}\eqref{part:inapprox}\\
                & \dmiPoSyCoAl & $W[2]$-hard when parameterized by~$d$ & \cref{thm:hard}\eqref{part:w[2]-compl}\\
                & \kmiPoSyCo & not solvable in~$2^{o(n)}$~time unless ETH fails& \cref{thm:hard}\eqref{part:eth-hard}\\
		\midrule
		
		\multirow{4}{*}{\rotatebox[origin=c]{90}{Sec.~\ref{sec:ccs}}}
                & \mpsc
                & solvable in   \(\ln 1/\varepsilon \cdot (4e^2/\sqrt{2\pi})^c \cdot 1/\sqrt{c}
  \cdot O(9^cm+4^cnm+nm\log n)\)
 time with error probability at most~$\varepsilon$
                & \cref{thm:ccs}\\
                & \mpsc & solvable in~$c^{O(c\log c)} \cdot n^{O(1)}$ time & \cref{thm:ccs}\\
                & \mpsc & solvable in~$O(3^n \cdot m)$ time & \cref{prop:exalg}\\
                & \kmiPoSyCo & WK[1]-hard parameterized by~$c$ & \cref{thm:wkhard}\\

		\bottomrule
	\end{tabularx}
\end{table}

\paragraph{Our contributions.}
Our work,
which is
summarized in \cref{tab:results},
is driven by the question when small input-specific parameter values allow for fast (exact) solutions in practically relevant special cases. 
Our ``use case scenarios'' herein are monitoring areas and
restoring connectivity in a sensor network after sensor faults.
In both scenarios, one naturally obtains lower bounds on the cost
paid by each vertex:
In the scenario of sensor faults,
one might want to reconnect the sensor network
without changing its topology too much,
that is,
so that each vertex pays at least the edges in an old solution.
In the scenario of monitoring areas,
we will notice that taking 
the cheapest edges incident to each vertex
into a solution
already yields a spanning subgraph
with few connected components.
The cost of this spanning subgraph coincides with the trivial lower bound
\[ L := \sum_{v \in V} \min_{\{u,v\}\in E} w(\{u,v\}) \]
and
it only remains to connect its components to get a solution.
When there is a solution whose cost coincides with the lower bound~$L$,
then it is easy to find: simply take for each vertex the incident edges
of minimum weight.
Naturally, the question arises whether
we can efficiently find a solution if its cost is only little more than~$L$?

In \cref{sec:hardness}, we show that the answer to this question is ``no''.
We show that \(o(\log n)\)\hyp approximating
the difference~\(d := \opt(G,w) - L\) between the minimum solution cost and~$L$ in an~$n$-vertex graph is NP-hard.
Assuming the Exponential Time Hypothesis (ETH) \citep{IP01},
also we show that there is no
exact \(2^{o(n)}\)-time algorithm for \miPoSyCo{}.  
Going even further,
we prove W[2]-hardness with respect to the parameter~$d$,
and thus,
under ETH,
there is no exact algorithm solving \miPoSyCo{} in \(f(d)\cdot n^{O(1)}\)~time
for any computable function~\(f\).

In \cref{sec:ccs}, we provide an (exact) algorithm for \miPoSyCo{} that exploits the above described lower bound~$L$ in a different way.
More precisely, we present an algorithm that works in polynomial time if one already has a sensor network with $O(\log n)$~connected components or if one can find a set of \emph{obligatory} edges that can be added to any optimal solution and yield a spanning subgraph with \(O(\log n)\) connected components.
In particular, this means that we show fixed-parameter tractability
for \miPoSyCo{} with respect to the parameter ``number~\(c\) of connected 
components in the spanning subgraph consisting of obligatory edges''.
Cases with small~\(c\) occur, for example,
in grid-like sensor arrangements,
which minimize
sensing area overlap
when monitoring 
areas~\citep{ZH05,ZEAC09}
or when about 10\,\% of all sensors
drop out of a triangular grid network (as we will see in \cref{sec:ccs-exp}).
We also show negative
results with respect to the parameter~$c$:
we show that the decision problem \kmiPoSyCo{}
is WK[1]-hard parameterized by~$c$;
problems that are WK[1]-hard parameterized by some parameter~$c$
are conjectured not to be reducible to solving instances of size~$c^{O(1)}$
\citep{HKS+15b}.

In \cref{sec:ccs-exp},
we conduct an experimental evaluation of our algorithm and
compare it to CPLEX on state-of-the art ILP models for \mpsc{},
one of which also exploits
the connected components of the
spanning subgraph consisting of obligatory edges.
We observe that our algorithms 
significantly outperform CPLEX with
the known ILP models
when $n$~is sufficiently compared to the fixed~$c$.

\section{Preliminaries}
\subsection{Notation}
We use \(\N\)~to denote the natural numbers
including zero.
By convention,
the maximum of the empty set is \(\max\emptyset=-\infty\)
and the minimum of the empty set is \(\min\emptyset=\infty\),
since these are neutral
elements with respect to taking the maximum
and the minimum, respectively.

\paragraph{Graph theory.}
We consider undirected,
finite,
simple graphs
$G=(V,E)$,
where $V$~is the set of \emph{vertices}
and $E \subseteq \{\{v,w\}\mid v\ne w\text{ and }v,w\in V\}$~is the set of \emph{edges}.
The \emph{(open) neighborhood} \(N_G(v):=\{u\in V\mid \{u,v\}\in E\}\) of a vertex~\(v\in V\)
is the set of vertices \emph{adjacent to~\(v\)} in~\(G\).
The \emph{closed neighborhood} of~\(v\)
is \(N_G[v]:=N_G(v)\cup\{v\}\).
For a vertex set~$U\subseteq V$, $G-U$ is the graph
obtained from~$G$ by deleting the vertices in~$U$
and their incident edges.
A \emph{clique} is a graph where each pair of vertices is adjacent.

\subsection{Computational model and complexity}
\label{sec:cc}
\paragraph{Computational model.}
We use the computational model
of the \emph{random access machine} (RAM)
with unit costs,
which acts on an array of rationals
and carries out
memory accesses, addition, subtraction,
multiplication, and division
in constant time \citep[Section~4.2]{Sch03}.
This assumption is justified
since our algorithms
will not operate on numbers
significantly larger than
the numbers given in the input.

\paragraph{Exponential time hypothesis.}
The Exponential Time Hypothesis (ETH) was introduced by \citet{IP01}
and it states that \textsc{3-Sat},
the problem of deciding the satisfiability
of a formula in 3-conjunctive normal form,
cannot be solved in~$2^{o(n+m)}$ time, where~$n$ and~$m$ are the number of variables and clauses in the input formula, respectively.
ETH implies FPT${}\ne{}$W[1] \citep{CHKX06}.

\paragraph{Inapproximability.}

In order to transfer inapproximability results of
some optimization problem~$\Pi$ to an optimization problem~$\Pi'$,
we will use $L$-reductions \citep{WS11}:
An \emph{$L$-reduction}
with parameters~$\alpha$ and~$\beta$
from a minimization problem~$\Pi$ to
a minimization problem~$\Pi'$ is a pair of
polynomial\hyp time algorithms~$A_1$ and $A_2$ such that
\begin{enumerate}[(i)]
\item $A_1$ turns any instance~$I$ of~$\Pi$
  into an instance~$I'$ of~$\Pi'$
  such that $\opt(I') \le \alpha \opt(I)$, and
  \label[part]{part:construction}
  \label[part]{part:upper-bound}
\item $A_2$ turns any solution of value~$\rho'$ for~$I'$
  into a solution of value~$\rho$ for~$I$
  such that \label[part]{part:lower-bound}
  $|\opt(I) - \rho| \le \beta|\opt(I') - \rho'|.$
\end{enumerate}
In particular,
if there is an $L$-reduction with $\alpha=\beta=1$
from $\Pi$ to $\Pi'$,
then
any $\gamma$-approximation for~$\Pi'$
transfers to a $\gamma$-approximation for~$\Pi$.
Thus,
if $\Pi$~is not $\gamma$-approximable in polynomial time,
then neither is~$\Pi'$.

\subsection{Fixed-parameter tractability}

\paragraph{Fixed\hyp parameter algorithms.}
The essential idea behind fixed\hyp parameter algorithms is
to accept exponential running times,
which are seemingly inevitable when solving NP-hard problems,
but to restrict them to one aspect of the problem,
the \emph{parameter} \citep{CFK+15,DF13,FG06,Nie06}.
Thus, formally, an instance of a \emph{parameterized problem}~\(\Pi\subseteq\Sigma^*\times\mathbb{N}\) is a pair~$(x,k)$ consisting of the \emph{input~$x$} and the \emph{parameter~$k$}.
A parameterized problem~$\Pi$ is \emph{fixed\hyp parameter tractable} with respect to a parameter~$k$ if there is an algorithm deciding~$(x,k)\in\Pi$ in $f(k) \cdot n^{O(1)}$~time for some computable function~$f$ and $n=|x|$.  Such an algorithm is called a \emph{fixed\hyp parameter algorithm}.
It is potentially efficient for small values of~$k$, in contrast to an algorithm that is merely running in polynomial time for each fixed~$k$.
FPT is the complexity class of fixed\hyp parameter tractable parameterized problems.
Observe that, for constant parameter values~$k$,
fixed-parameter tractability implies polynomial-time solvability,
where the degree of the
polynomial is \emph{independent} of~$k$.
XP is the complexity class of parameterized problems solvable in polynomial time if the parameter is a constant, thus allowing for parameter dependencies in the degree of the running-time polynomial.

\paragraph{Parameterized intractability.}
The parameterized analog of P $\subseteq$ NP is
a hierarchy of complexity classes
FPT\({}\subseteq{}\)W[1]\({}\subseteq{}\)W[2]\({}\subseteq{}\ldots\subseteq{}\)XP,
where all inclusions are assumed to be proper.
A~parameterized problem~$\Pi$ with parameter~\(k\) is
\emph{W[$t$]-hard} for some \(t\in\mathbb{N}\)
if every problem in W[$t$] has a parameterized reduction to~\(\Pi\):
a \emph{parameterized reduction} from a parameterized problem~$\Pi_1$
to a parameterized problem~$\Pi_2$ is an algorithm mapping an instance~$(x,k)$
to an instance~$(x',k')$
in~$f(k)\cdot|x|^{O(1)}$~time
such that $k'\leq g(k)$ and
\((x,k)\in\Pi_1\iff(x',k')\in\Pi_2\),
where \(f\)~and~\(g\) are arbitrary computable functions.
If $g$~is a polynomial,
then the parameterized reduction is called a \emph{polynomial parameter transformation (PPT)}.
By definition,
no W[$t$]-hard problem is fixed\hyp parameter tractable
unless FPT${}={}$W[$t$].

\subsection{(Hardness of) provably effective data reduction}
\paragraph{Kernelization.}
Kernelization is the main formalization
of data reduction with provable performance guarantees
\citep{FLSZ19}.
It has also proven effective in
experimental studies \citep{ABN06,BS20,BFTxx,MPMM10}.

A~\emph{kernelization} for a parameterized problem~$\Pi\subseteq\Sigma^*\times\mathbb N$
is a polynomial\hyp time algorithm
that maps any instance~$(x,k)\in\Sigma^*\times \mathbb N$
to an instance~$(x',k')\in\Sigma^*\times\mathbb N$
such that
$(x,k)\in \Pi \iff {(x',k')\in \Pi}$, and
 $|x'|+k'\leq g(k)$ for some computable function~\(g\).
We call \((x',k')\) the \emph{problem kernel}
and \(g\) its \emph{size}.
A generalization of problem kernels
are \emph{Turing kernels},
where one is allowed to generate multiple reduced
instances instead of a single one:
A~\emph{Turing kernelization} for
a parameterized problem~$\Pi\subseteq\Sigma^*\times\mathbb N$
  is an algorithm that decides~$(x,k)\in \Pi$
  in polynomial time
  given access to an oracle
  that answers $(x',k')\in \Pi$ in constant time
  for any \((x',k')\in\Sigma^*\times\mathbb N\) with
  $|x'|+k'\leq g(k)$,
  where $g$~is an arbitrary function
  called the \emph{size} of the Turing kernel.

  \paragraph{Kernelization hardness.}
  \emph{WK[1]-hard} parameterized problems
  with parameter~$k$
  do not have problem kernels of size polynomial in~$k$
  unless the polynomial\hyp time hierarchy collapses
  and are conjectured not to have
  Turing kernels of polynomial size in~$k$
  either~\citep{HKS+15b}.
  Herein,
  a problem~$\Pi$ is \emph{WK[1]-hard}
  if every problem in WK[1] has
  a polynomial parameter transformation to~$\Pi$.
  An example for a WK[1]-hard problem is
  Set Cover parameterized by the size of the universe.

\section{Parameterization above lower bound on total cost}\label{sec:hardness}
In this section, we show that lower bounds on the total costs are hard to exploit algorithmically.
To this end,
we formalize the idea of lower bounds as follows.
\begin{definition}[\plb{}s] \label{def:plb}
  \emph{\Plb{}s} are
  given by a
  function~\({\ell\colon V\to\N}\)
  such that
  there is an optimal solution~\(T=(V,F)\)
  to \miPoSyCo{} satisfying
  \[
    \max_{\{u,v\}\in F} w(\{u,v\})\geq \ell(v)\quad\text{  for every vertex~\(v\in V\)}.
  \]
\end{definition}

\begin{example}\label{ex:nn}
  \looseness=-1
  A trivial \plb{}~\(\ell(v)\) is given by the weight of the minimum\hyp weight edge incident to~\(v\), because \(v\)~has to be connected to some vertex in any solution.
  A more sophisticated \plb{} is given by
  \[
    \ell(v)=\max_{G'\in C}\min_{u\in V(G')}w(\{u,v\}),\text{\quad where $C$~is the set of connected components of~$G-\{v\}$.}
  \]
\end{example}
\noindent
In the rest of this section,
we concentrate on
the trivial lower bound
given by the minimum\hyp weight edge.
Thus, the overall cost of a solution is at least
\[ L = \sum_{v \in V} \ell(v) = \sum_{v \in V} \min_{\{u,v\}\in E} w(\{u,v\}). \]
If the weights~\(w\colon E\to\N\) of the edges in a graph~\(G=(V,E)\) are at least one, then this immediately yields a ``large'' lower bound~$L\geq n$ on the cost of an optimal solution.
This implies that even constant-factor approximation algorithms (e.\,g., the one by \citet{ACM+06}) can return solutions that are, in absolute terms, quite far away from the optimum. 
Furthermore,
fixed-parameter tractability for \kmiPoSyCo{}
parameterized by the solution cost
immediately follows from \cref{prop:exalg} (in \cref{sec:outl}).

A more desirable and stronger result
would be about the difference~$d$ between the optimal solution cost and~$L$:
for example a polynomial\hyp time constant\hyp factor approximation of~$d$
or a fixed-parameter tractability result with respect to the parameter~$d$~\citep{MR99,CPPW13,GP16b,BFK18}.
However,
under the Exponential Time Hypothesis,
we show that such algorithms do not exist.
To formally state our hardness results, we use the following problem variant,
which incorporates the lower bound.

\optprob{\miPoSyCo{} Above Lower Bound (\miPoSyCoAl{})}%
{\label{prob:mpscal}A connected undirected graph~$G = (V,E)$ and edge weights~$w\colon E \to \N$.}%
{Find a connected spanning subgraph~$T=(V,F)$ of~$G$ that minimizes
  \begin{align}
    \sum_{v \in V}
    \smashoperator[r]{\max_{\{u,v\}\in F}}w(\{u,v\}) -
    \sum_{v \in V}
    \smashoperator[r]{\min_{\{u,v\}\in E}}w(\{u,v\}). \label{goalfunc-AL}
  \end{align}
}
\noindent
For a \miPoSyCoAl{} instance~$I=(G,w)$,
we denote by~$\mar(I)$ the minimum value of \eqref{goalfunc-AL}
(we also refer to~$\mar(I)$ as the \emph{\alValue{}} of~$I$).
For showing hardness results,
we will also consider the \emph{decision version} of the problem:
By \dmiPoSyCoAl{},
we denote the problem of deciding whether
an \miPoSyCoAl{} instance~$I=(G,w)$ satisfies $\mar(I)\leq d$.

\pagebreak[3]
\begin{theorem}
  \label{thm:hard}
  The following properties hold even in graphs with edge weights one and two:
  \begin{enumerate}[(i)]
  \item \miPoSyCoAl{} is NP-hard to approximate within a factor of~$o(\log n)$, \label{part:inapprox}
  \item \dmiPoSyCoAl{} is W[2]-hard when parameterized by~$d$, and\label{part:w[2]-compl}%
  \item unless ETH fails, \kmiPoSyCo{} and \dmiPoSyCoAl{} are not solvable in~$2^{o(n)}$ time. \label{part:eth-hard} %
  \end{enumerate}
  Moreover,
  \dmiPoSyCoAl{}
  is solvable in
  in $O(2^d\cdot n^{d+2})$~time
  and 
  \eqref{part:inapprox} and \eqref{part:w[2]-compl} also hold in
  complete graphs with metric edge weights.
\end{theorem}

\noindent
\cref{thm:hard}\eqref{part:w[2]-compl} implies that,
under ETH,
\dmiPoSyCoAl{} is not solvable in \({f(d)\cdot n^{O(1)}}\) time \citep{CHKX06}.
In contrast,
\cref{thm:hard} also shows that
it is solvable in polynomial time for any fixed~\(d\).
We prove \cref{thm:hard}
using a modified reduction from \MSC{} to \miPoSyCo{} due to \citet{EPS13}.

\optprob{\MSC{}}%
{A universe~$U$ and a family~$\mathcal{F}$ of subsets of~$U$.}%
{Find a \emph{set cover}~$\mathcal{F'} \subseteq \mathcal{F}$ (that is, $\bigcup_{S\in\mathcal{F'}}S = U$) of minimum size.}

\noindent
We now present the reduction of \citet{EPS13}
and then show that a modification yields \cref{thm:hard}.
Note that Erzin et al. used edge weights zero and one in their reduction,
whereas we use positive integers,
namely one and two.

\begin{transformation}\label{trans:phi}
  Given an instance~$(U, \mathcal{F})$ of \MSC{},
  construct an instance~$(G, w)$ of \miPoSyCoAl{} as follows.
  The vertex set of the graph~$G=(V,E)$
  is $V:= \{s\} \uplus U \uplus \mathcal F$
  for some new vertex~$s$.
  There is an edge~$\{u,S\}$ of weight two in~$G$
  for each~$u\in U$ and~$S \in \mathcal F$
  such that~$u \in S$. 
  Moreover, there is an edge~$\{s, S\}$
  of weight one for each~$S\in\mathcal F$.
\end{transformation}

\noindent
\cref{trans:phi} creates graphs with
edge weights one and two.
The following transformation
completes it to satisfy the triangle inequality.

\begin{transformation}
  \label{trans:metric}
  Given an instance~$(U, \mathcal{F})$ of \MSC{}, %
  first apply \cref{trans:phi}
  and,
  in the resulting instance~$(G,w)$ of \mpsc{},
  turn~$G$ into a complete graph,
  assigning all newly added edges~$\{u,v\}$
  a weight equal to the length of a shortest weighted $u$-$v$-path.
\end{transformation}

\noindent
\cref{trans:metric} introduces
edges of weight three and four.
To prove the correctness of the reduction,
we first prove that it is never required
for these edges to be used in a solution.

\begin{lemma}
  \label[lemma]{lem:metstruct}
  Let $(U,\mathcal F)$~be an instance of \SC{} and
  $(G,w)$ be an \mpsc{} instance
  generated from it by \cref{trans:phi} or \cref{trans:metric},
  where $G=(V,E)$ and $V=\{s\}\uplus U\uplus\mathcal F$.
  Finally, for a solution~$T=(V,F)$ to~$(G,w)$,
  let
  \[
    p_T(v):=\max_{\{u,v\}\in F}w(\{u,v\})
  \]
  be the cost of vertex~$v\in V$ paid in~$T$.
  Then,
  the cost of~$T$ is at least~$1+|\mathcal F|+2|U|$ and,
  in polynomial time
  without increasing its cost, one can transform~$T$
  so that
  \begin{enumerate}[(i)]
  \item\label{metstruct1} $p_T(s)=1$,
  \item\label{metstruct2} $p_T(S)\in\{1,2\}$ for all~$S\in\mathcal F$,
  \item\label{metstruct3} $p_T(u)=2$ for all~$u\in U$.
  \end{enumerate}
\end{lemma}

\begin{proof}
Take an arbitrary
solution~$T=(V,F)$ for~$(G,w)$.
Obviously, $p_T(s)\geq 1$, $p_T(u)\geq 2$ for any $u\in U$,
and $p_T(S)\geq 1$ for any $S\in\mathcal F$,
yielding the lower bound of $1+|\mathcal F|+2|U|$
on the cost of~$T$.

Since the rest of the lemma obviously holds for \cref{trans:phi}, we prove it for \cref{trans:metric}.
To this end,
we partition~$\mathcal F=\mathcal F_1\uplus \mathcal F_2$
so that $p_T(S)=1$ for all $S\in \mathcal F_1$
and $p_T(S)\geq 2$ for all $S\in \mathcal F_2$.
Moreover,
consider $U_2=\bigcup_{S\in \mathcal F_2} S$ and $U_1=U\setminus U_2$.
Observe that $p_T(u)\geq 3$ for each~$u\in U_1$,
since there are no edges of weight one incident to~$u$
and each of $u$'s neighbors along
an edge of weight two pays only cost one.
Thus,
the cost of~$T$ is
\begin{align}
  \notag\sum_{v\in V}p_T(v)&\geq |\{s\}|+|\mathcal F_1|+2|\mathcal F_2|+2|U_2|+3|U_1|\\&=1+|\mathcal F|+2|U|+|\mathcal F_2|+|U_1|.\label{Tlobo}
\end{align}
We now describe a solution~$T'$
satisfying \cref{lem:metstruct}\eqref{metstruct1}--\eqref{metstruct3}
that can be computed in polynomial time from~$T$
and whose cost is exactly \eqref{Tlobo}.
\begin{itemize}
\item $T'$ contains all edges of weight one incident to~$s$.
  These connect all vertices~$S\in\mathcal F$ to~$s$.
  Note that $1=p_{T'}(s)\leq p_T(s)$.
  
\item $T'$ contains all edges of weight two incident
  to some~$S\in \mathcal F_2$.
  These connect all vertices~$u\in U_2$
  to~$\mathcal F_2$ and, hence, to~$s$.
  Note that $2=p_{T'}(S)\leq p_{T}(S)$ for $S\in\mathcal F_2$.
  Since $T'$ will not contain any other edges incident to~$U_2$,
  $2=p_{T'}(u)\leq p_{T}(u)$ for any $u\in U_2$.
  
\item It remains to connect each vertex~$u\in U_1$ to the rest of the graph.
  To this end,
  $T'$ contains an arbitrary edge~$\{u,S\}$
  such that $u\in S$, which has weight two.
  Note that $2=p_{T'}(u)<3\leq p_T(u)$ for any $u\in U_1$.
\end{itemize}

\noindent
Let $K\subseteq V$ be the set of vertices~$u$
such that $p_{T'}(u)> p_{T}(u)$.
That is,
actually, $K\subseteq\mathcal F_1$
and $|K|\leq |U_1|$,
since $K\subseteq \mathcal F_1$ consists exactly
of the end points
of the edges in~$T'$ we added for each $u\in U_1$.
Moreover, $p_{T'}(S)=2$ for each~$S\in K$.
Thus,
the cost of this solution~$T'$ is
\begin{align*}
  \sum_{v\in V}p_{T'}(v)&=|\{s\}|+|\mathcal F_1\setminus K|+2|K|+2|\mathcal F_2|+2|U|\leq 1+|\mathcal F_1|+|U_1|+2|\mathcal F_2|+2|U|=\eqref{Tlobo}.\qedhere
\end{align*}
\end{proof}

\noindent
The following lemma,
together with the approximation hardness of \MSC{},
will yield \cref{thm:hard}\eqref{part:inapprox}.

\begin{lemma}\label{lem:L-reduction}
  \cref{trans:phi} yields an
  $L$-reduction with parameters~$\alpha = \beta = 1$ from \MSC{} to \miPoSyCoAl{} in graphs with edge weights one and two,
  whereas \cref{trans:metric}
  yields such an $L$-reduction to \miPoSyCoAl{}
  in complete graphs with metric edge weights.
\end{lemma}

\begin{proof}
If is obvious that \cref{trans:phi,trans:metric}
create graphs with edge weights one and two
and metric complete graphs, respectively.
We verify that they satisfy properties (\ref{part:construction}) and (\ref{part:lower-bound}) of $L$-reductions (cf.\ \cref{sec:cc}).

\eqref{part:construction}
Both transformations
work in polynomial time.
It remains to show that the \alValue{} of the optimal solution for
the instance~$I'=(G,w)$ of \miPoSyCoAl{}
is at most the cost of the optimal solution for
the instance~$I=(U,\mathcal F)$ of \MSC{}.
To this end, let~$\mathcal{F'} \subseteq \mathcal{F}$ be a set cover of minimum size and thus~$\opt(I) = |\mathcal{F'}|$.
Consider the spanning subgraph~$T$ of~$G=(\{s\}\uplus U\uplus\mathcal F, E)$ %
that contains all edges of weight one incident to~$s$
and
all edges~$\{u,S\}$ (of weight two) such that
$S \in \mathcal{F'}$ and~$u \in S$.
Then $T$~is a connected spanning subgraph: 
All vertices in~$\mathcal{F}$ are connected via~$s$
by the edges with weight one. 
Furthermore, each vertex in~$U$ is connected to at least one vertex in~$\mathcal{F'}$ since~$\mathcal{F'}$ is a set cover.
We analyze its margin to show~$\mar(I') \le \opt(I)$.

By \cref{lem:metstruct}, the lower bound of the cost of any connected spanning subgraph in~$G$ is~$2|U|+|\mathcal F|+1$.
Yet the overall cost of~$T$ is at most~$2|U|+|\mathcal F|+1+\opt(I)$
and thus the \alValue{}~$\rho'$ of~$T$ is~$\opt(I)$:
Compared to the lower bound \(2|U|+|\mathcal F|+1\),
each vertex~$S$ with $S\in \mathcal{F'}$
pays one additionally (in total two)
since it is incident to a weight-two edge in~\(T\)
that has the other endpoint in~$U$.
Thus, the overall cost is~$2|U| + |\mathcal F| + 1 +\opt(I)$
and~$\mar(I') \le \opt(I)$.

\eqref{part:lower-bound}
We show how to transform solutions
for the instance~$I'=(G,w)$ of \miPoSyCoAl{}
into set covers for~$I=(U,\mathcal F)$.
To this end, let $T$~be a connected spanning subgraph of~$G=(\{s\}\uplus U\uplus\mathcal F, E)$.
Without loss of generality,
$T$ is of the form described by \cref{lem:metstruct}:
$s$ pays one,
each vertex~$u\in U$ pays two,
and each vertex~$S \in \mathcal{F}$ pays one or two.
We show that the subset~$\mathcal F' \subseteq \mathcal{F}$
of vertices paying two is a set cover for~$I$.
To this end, consider any element~\(u\) in~\(U\). %
Since~$T$ is connected,
it follows that at least one edge~$e$
is incident to~$u$ is in~$T$.
Since $w(e)=2$,
by construction, the second endpoint of~$e$
is in~$\mathcal{F}$ and thus,
$e = \{u, S\}$ for some~$S \in \mathcal{F}$ with $u\in S$.
By definition of~$\mathcal{F'}$,
we have~$S \in \mathcal{F'}$.
Thus, \(u\)~is covered by~\(\mathcal{F'}\).

\looseness=-1
Observe that the \alValue{} of~$T$ is~$\rho = |\mathcal F'|$,
which is equal to the size~$\rho'$ of the set cover~$\mathcal{F'}$.
Since this argument holds for any connected spanning subgraph of~$G$,
we have~$\opt(I) \le \mar(I')$.
Since we already showed~$\mar(I') \le \opt(I)$
it follows that~$\mar(I') = \opt(I)$.
Hence, we arrive at~$|\opt(I)-\rho| = |\mar(I') - \rho'|$.
\end{proof}

\noindent
The following lemma,
together with the W[2]-hardness of $k$-\SC{},
will yield
\cref{thm:hard}(\ref{part:w[2]-compl}).

\begin{lemma}\label{lem:param-reduction}
  \cref{trans:phi,trans:metric}
  yield parameterized reductions from \(k\)-\SC{} parameterized by the solution size~\(k\) to $d$-\miPoSyCoAl{} parameterized by~$d$
  in graphs with edge weights one and two
  or in complete graphs with triangle inequality,
  respectively.
  The graph constructed by \cref{trans:phi}
  has $O(|U| + |\mathcal{F}|)$~vertices.
\end{lemma}
\begin{proof}
We slightly enhance \cref{trans:phi,trans:metric} using the decision versions of the problems and setting~$d := k$.
Note that, in the proof of \cref{lem:L-reduction},
we showed  that there is a set cover of size~$k$ in the given \MSC{} instance if and only if there is a connected spanning subgraph~$T$ with cost~$2|U|+|\mathcal F|+1+d$ in the constructed $d$-\miPoSyCoAl-instance.
Finally, observe that \cref{trans:phi}
creates a graph where the number of vertices
is~$O(|U|+|\mathcal F|)$. 
\end{proof}

\noindent
Combining the known intractability of \(k\)-\SC{} with \cref{lem:L-reduction,lem:param-reduction}, we can finally prove \cref{thm:hard}.

\begin{proof}[Proof of \cref{thm:hard}.]
(\ref{part:inapprox}) follows from \cref{lem:L-reduction} and the fact that \MSC{} is NP-hard to approximate within a factor of~$o(\log n)$~\citep{RS97}.

(\ref{part:w[2]-compl}) follows from \cref{lem:param-reduction} and the fact that \(k\)-\SC{} is W[2]-complete parameterized by the solution size~$k$~\citep{DF13}.

(\ref{part:eth-hard}) follows from \cref{lem:param-reduction}, the observation that \cref{trans:phi} runs in polynomial time, and the fact that \(k\)-\SC{} cannot be solved in~$2^{o(|U|+|\mathcal F|)}$~time unless the ETH fails~\citep{IPZ01}.

Finally,
\mpsc{} can be solved in
in $O(2^d\cdot n^{d+2})$~time
as follows:
At most \(d\)~vertices can pay more than their \plb{}.
We can try all possibilities for choosing \(i\leq d\)~vertices, all \({d\choose i}\)~possibilities to increase their total cost by at most~\(d\), and check whether the graph of the ``paid'' edges is connected.
The algorithm runs in \(\sum_{i=1}^d{n\choose i}{d\choose i}\cdot O(n+m)\subseteq O(2^d\cdot n^{d+2})\)~time.
\end{proof}

\section{Parameterizing by the number of connected components induced by obligatory edges}
\label{sec:ccs}
Complementing our hardness results in \cref{sec:hardness}, we show in
this section how \plb{}s (see \cref{def:plb}) can be algorithmically
exploited.
To this end,
in \cref{sec:obl},
we describe how \plb{}s
induce, what we call, an \emph{obligatory subgraph},
which is part of at least one optimal solution.

In \cref{sec:outl},
we present the algorithm
that solves \mpsc{} efficiently
if the obligatory subgraph has few connected components.
In \cref{sec:corr,sec:dp,sec:rt},
we prove its correctness and running time.

Finally,
in \cref{sec:dr},
we prove that it is hard to reduce
arbitrary instances of \mpsc{}
to equivalent instances
with a size polynomial in the number
of the connected components of the obligatory subgraph.

\looseness=-1
Notably,
the obligatory subgraph does not necessarily have to be given
by \plb{}s.
It can also arise as part of an application scenario
and given as input,
for example, when reconnecting a sensor network
that has lost connectivity due to sensor faults,
as studied by \citet{RM99}. 

\subsection{Finding obligatory edges}
\label{sec:obl}

To find obligatory edges, we use \plb{}s (see \cref{def:plb}).
Once we have \plb{}s,
we can compute an \emph{obligatory subgraph},
whose edges are contained in at least one optimal solution.

\begin{definition}[obligatory subgraph]\label{def:oblig}
  The \emph{obligatory subgraph~\(G_\ell\)}
  of a graph~\(G=(V,E)\)
  induced by
  \plb{}s~\(\ell\colon V\to\N\)
  consists
  of all vertices of~\(G\) and all \emph{obligatory edges} \(\{u,v\}\) with
  \[\min\{\ell(u),\ell(v)\}\geq w(\{u,v\}).\]
\end{definition}

\noindent
The better the \plb{}s~$\ell$ are,
the more obligatory edges they potentially induce,
thus reducing the number~\(c\) of connected components
of~\(G_\ell\).
Clearly, coming up with good \plb{}s is a challenge on its own.
Yet already the simple \plb{}s in \cref{ex:nn}
may yield obligatory subgraphs
with few connected components in applications:

\begin{example}
  \label[example]{ex:broken-grids}
  Consider the vertex lower bounds~\(\ell\) from \cref{ex:nn}.
  If we arrange sensors in a grid,
  which is the most energy-efficient arrangement of sensors
  for monitoring areas \citep{ZH05,ZEAC09},
  then \(G_\ell\)~has only one connected component.
  The number of connected components
  may increase due to sensor defects
  that disconnect the grid
  or
  due to varying sensor distances within the grid.
  The worst case
  is if the sensors have pairwise distinct
  distances.
  Then,
  \(G_\ell\) might have only one edge,
  joining the closest pair, and
  \(n-1\)~connected components.
\end{example}

\noindent
Alternatively,
an obligatory subgraph may arise
as the result of a sensor network
that lost connectivity due to faulty
sensors and has to be reconnected
at minimum extra cost.

\subsection{A fixed\hyp parameter algorithm}
\label{sec:outl}
The number~\(c\) of connected components in~\(G_\ell\)
can easily be  exploited in an exact \(O(n^{2c})\)-time
algorithm for \miPoSyCo{},\footnote{To connect the \(c\)~components of~\(G_\ell\),
  one has to add \(c-1\)~edges.
  These have at most \(2c-2\)~endpoints.
  One can try all \(n^{2c-2}\)~possibilities
  for choosing these endpoints and
  check each resulting graph for connectivity in \(O(n+m)\subseteq O(n^2)\)~time.
  Altogether, this proves containment of \miPoSyCo{} parameterized by~$c$ in the class~XP.
  }
which runs in polynomial time for constant~\(c\),
yet is inefficient already for small values of~\(c\).
We develop, among others, a randomized algorithm
that runs in polynomial time even for \(c\in O(\log n)\):
\begin{theorem}\label{thm:ccs}
  \miPoSyCo{} with \plb{}s~\(\ell\) is solvable
  \begin{enumerate}[(i)]
  \item\label{thm:ccs1} in
    \(O(\ln 1/\varepsilon \cdot (4e^2/\sqrt{2\pi})^c
  \cdot (9^cm+4^cnm+nm\log n))\)~time by a randomized algorithm with an error probability at most~\(\varepsilon\)
    for any given~$\varepsilon$ with \(0<\varepsilon<1\), and
  \item\label{thm:ccs2} in \(c^{O(c\log c)}\cdot n^{O(1)}\)~time by a deterministic algorithm, 
  \end{enumerate}
  where \(c\)~is the number of connected components of the obligatory subgraph~\(G_\ell\).
\end{theorem}

\begin{remark}
The algorithms behind \cref{thm:ccs} work for \emph{any} \plb{}s~$\ell$.
Thus, any (heuristical) approach improving on the \plb{}s will directly (and provably) improve the performance of the algorithms.

The deterministic algorithm
in \cref{thm:ccs}\eqref{thm:ccs2}
is primarily of theoretical interest
because it classifies \miPoSyCo{}
as \emph{fixed-parameter tractable} parameterized by~\(c\).
Practically,
the randomized algorithm
in \cref{thm:ccs}\eqref{thm:ccs1}
is much easier to implement
and it is the one that we experimentally
evaluate in \cref{sec:ccs-exp}.

Note that the parameter~$c$ is upper-bounded by~$n$.
Thus, \cref{thm:hard}\eqref{part:eth-hard} also implies that,
under ETH, there is no~$2^{o(c)} (n+m)^{O(1)}$-time algorithm for \miPoSyCo{},
whereas our randomized algorithm in \cref{thm:ccs}\eqref{thm:ccs1} has running time~$2^{O(c)} n^{O(1)}$ for constant error probability~$\varepsilon$.

The number of connected components
of obligatory subgraphs
has recently also been exploited
in fixed-parameter algorithms for problems
of servicing links in transportation networks
\citep{GWY17,SBNW11,SBNW12,BFTxx},
which led to practical results \citep{BKS17,BFTxx}.
\end{remark}

\noindent
We will now prove \cref{thm:ccs}.
The proof also yields the following
deterministic algorithm for \miPoSyCo{}.
It is much faster than
the trivial algorithm
enumerating all of the possibly \(n^{n-2}\)~spanning trees.
Moreover,
\Cref{thm:hard}\eqref{part:eth-hard} shows that
it is asymptotically optimal:

\begin{proposition}\label{prop:exalg}
  \miPoSyCo{} can be solved in \(O(3^{n}\cdot m)\)~time.
\end{proposition}

\noindent
Like some known approximation algorithms
for \mpsc{} \citep{HW16,ACM+06},
our algorithms in \cref{thm:ccs}
work by adding edges to~\(G_\ell\)
in order to connect its \(c\)~connected components.
In contrast to the known approximation algorithms,
our algorithms will find
an \emph{optimal} set of edges to add.
We describe our algorithm
in terms of a \emph{padded} version~\(\pG{}\) of the input graph~\(G\),
in which each connected component of~\(G_\ell\) is turned into a clique.
In the padded graph,
it is sufficient to search for connected subgraphs of~\(\pG{}\)
that contain at least one vertex of each connected component
of~\(G_\ell\):
We can then add the edges in~\(G_\ell\) to such subgraphs
in order to obtain a connected spanning subgraph of~\(G\).
This simplifies the problem.

\begin{definition}[padded graph, components]\label{def:padded}
  Let \(\ell\colon V\to\N\) be \plb{}s
  for a graph~\(G=(V,E)\)
  with edge weights~\(w\colon E\to\N\).
  We denote the \(c\)~connected components
  of the obligatory subgraph~\(G_\ell\) by~\(G_\ell^1,G_\ell^2,\dots,G_\ell^c\).
    
  The \emph{padded graph}~\(\pG{}=(V,\pE{})\)
  is a supergraph of~\(G\) with the edge set~\(\pE=E\cup \ppE\), where
  \[
    \ppE:=\{\{u,v\}\subseteq V\mid \text{$u$ and $v$ are in the same component of~\(G_\ell\)}\},
  \]
  and edge weights
  \[
    \pw\colon\pE{}\to\N,\{u,v\}\mapsto
    \begin{cases}
      0&\text{ if }\{u,v\}\in\ppE\text{ and}\\
      w(\{u,v\})&\text{ otherwise}.
    \end{cases}
  \]
\end{definition}

\paragraph{Algorithm outline.}
To solve a \mpsc{} instance~\((G,w)\)
with \plb{}s~\(\ell\colon V\to\N\),
we have to add \(c-1\)~edges to~\(G_\ell\)
in order to connect its \(c\)~connected components.
These edges have at most \(2c-2\)~endpoints.
Thus,
we need to find a %
connected subgraph in~\(\pG{}\) that
\begin{itemize}
	\item contains at most \(2c-2\)~vertices,
	\item contains at least one vertex of each connected component of~\(G_\ell\), and
	\item minimizes the total
          cost increase compared %
          to
          the vertex lower bounds $\ell\colon V\to\N$.
\end{itemize}
We will do this using the color-coding technique
introduced by \citet{AYZ95}:
randomly color the vertices of~\(\pG{}\)
using at most \(2c-2\)~colors
and then search for connected subgraphs of~\(\pG{}\)
that each contain exactly one vertex of each color
and in which the total cost increase
compared to the vertex lower bounds~$\ell\colon V\to\N$
is minimized.
Formally,
we will solve the following problem on~\(\pG{}\).

\optprob{Min-Power Increment Colorful Connected Subgraph (\MCCS{})}%
{\label[problem]{prob:mcpu}A connected undirected graph~$G = (V,E)$, edge weights~$w\colon E \to \N$, vertex colors~\(\col\colon V\to\N\), a function~\(\ell\colon V\to\N\), and a color subset~\(C\subseteq\N\).}%
{Compute a connected subgraph~\(T=(W,F)\) of~\(G\) such that
  \(\col\) is a bijection between \(W\) and~\(C\) and such that $T$~minimizes
  \begin{equation}
    \sum_{v \in W} \max \Bigl\{ 0, \max_{\{u,v\}\in F}w(\{u,v\})-\ell(v) \Bigr\}.\label{eq:mccs}
  \end{equation}
}

Note that,
in order to connect the components of~$G_\ell$
using \MCCS{},
we cannot simply color
the vertices of the input graph~\(G\)
\emph{completely} randomly:
One component of~\(G_\ell\) could
contain all colors and, thus,
a connected subgraph containing all colors
does not necessarily connect the components of~\(G_\ell\).
Instead,
we employ a trick that was previously applied
mainly heuristically in algorithm engineering in order to
increase the success probability of color-coding algorithms \citep{BBF+11,BHK+10,DSG+08}:
Since we know that our sought subgraph
contains at least one vertex of each connected
component of~\(G_\ell\),
we color the connected components of~\(G_\ell\)
using
pairwise disjoint color sets.
Herein,
we first ``guess'' the number~\(\nc_i\) of vertices
that the sought subgraph will contain
of each connected component~\(G_\ell^i\) of~\(G_\ell\)
and use \(\nc_i\)~colors to color each component~\(G_\ell^i\).
We thus arrive at the following algorithm for \mpsc{}:

\begin{enumerate}
	\item construct \(\pG{}\) (in which each connected component of~$G_\ell$ is a clique)
	\item Repeat a certain number of times:
	\begin{enumerate}
        \item color the vertices randomly so that
          the connected components of~$G_\ell$
          get pairwise disjoint colors.
        \item solve \MCCS{} on the colored graph.
	\end{enumerate}
      \end{enumerate}

      \noindent
It is formalized as follows:

\begin{algorithm}[for \mpsc{}]  \label[algorithm]{alg:ccs}
  \leavevmode

  \noindent
  \textit{Input:}
    A \mpsc{} instance~\(I=(G,w)\),
    \plb{}s~\(\ell\colon V\to\N\) for~\(G=(V,E)\), and an upper bound~\(\varepsilon\) on the error probability,
    where $0<\varepsilon<1$.

    \noindent
    \textit{Output:}
    With probability at least~\(1-\varepsilon\)
    an optimal
    solution for~\(I\).

  \begin{compactenum}
  \item \(c\gets{}\)number of connected components of the obligatory subgraph~\(G_\ell\).
  \item\label{lin:comploop} \textbf{for each} \(\nc_1,\nc_2,\dots,\nc_c\in\N\setminus\{0\}\) such that \(\sum_{i=1}^c\nc_i\leq 2c-2\) \textbf{do}:
    
  \item\label{lin:colchoice} \quad choose pairwise disjoint~\(C_i\subseteq\{1,\dots,2c-2\}\)
    with \(|C_i|=\nc_i\) for \(i\in\{1,\dots,c\}\).
  \item\label{lin:repeat} \quad\textbf{repeat} \(t:=\lceil\ln\varepsilon/\ln(1- \prod_{i=1}^c\nc_i!/\nc_i^{\nc_i})\rceil\) \textbf{times}:%
    \footnote{Repeat once if the logarithm is undefined, that is,
      if \(\nc_i=1\) for all \(i\in\{1,\dots,c\}\).}
    
  \item\label{lin:randcol} \qquad \textbf{for each} \(i\in\{1,\dots,c\}\), randomly color the vertices of component \(G_\ell^i\) of~\(G_\ell\)

    \qquad\quad using colors from~\(C_i\), let the resulting coloring be \(\col\colon V\to\N\).
  \item\label{lin:onemcpusol} \qquad Solve \MCCS{} instance \(\pI:=(\pG,\pw,\col,\ell,C)\) as described in
    \cref{sec:dp}.
    
  \item\label{lin:mcpusol} let \(T=(W,F)\) be the best \MCCS{} solution found in any of the repetitions.
    
  \item\label{lin:consF}\label{lin:retgraph} \textbf{return} \(T'=(V,(F\setminus \ppE)\cup E_\ell)\).
  \end{compactenum}
\end{algorithm}

\noindent
We prove the correctness of \cref{alg:ccs}
in \cref{sec:corr}.
Then, in \cref{sec:dp},
we show how to solve the \MCCS{} instance
in line~\ref{lin:onemcpusol} of~\cref{alg:ccs}.
Finally,
in \cref{sec:rt},
we analyze the running time of
\cref{alg:ccs}
and also show how to derandomize it
to complete the proof of \cref{thm:ccs}.
\subsection{Correctness of \cref{alg:ccs}}
\label{sec:corr}
We will now prove the correctness of \cref{alg:ccs}.
First, with the following lemma,
we prove that,
if \cref{alg:ccs}
chooses a suitable coloring
in line~\ref{lin:randcol},
then the \MCCS{} instance~\(\pI\)
solved in line~\ref{lin:onemcpusol}
has a solution
of cost at most \(\opt(I)-\sum_{v\in V}\ell(v)\).

\begin{lemma}\label{lem:mpsc-mcpu}
  Let \(I=(G,w)\)~be a \miPoSyCo{} instance,
  let \(\ell\colon V\to\N\)~be
  \plb{}s
  for~\(G=(V,E)\),
  and let
  \(c\)~be the number of connected components of~\(G_\ell=(V,E_\ell)\).
  Then, there is an optimal solution~\(T=(V,F)\) for~\(I\)
  such that
    \begin{enumerate}[(i)]
    \item the set~\(W\subseteq V\) of vertices
      incident to an edge in~\(F\setminus E_\ell\)
      has at most \(2c-2\)~vertices, and
    \item for \(C=\{1,\dots,|W|\}\)
      and any coloring \(\col\colon V\to C\)
      inducing a bijection between~\(W\) and~\(C\),
      there is a solution~\(T'=(W,F')\)
      to the \MCCS{} instance~\((\pG{},\pw{},\col,\ell,C)\)
      with cost at most
      \[\opt(I)-\smashoperator{\sum_{v\in V}}\ell(v).\]
    \end{enumerate}
\end{lemma}

\begin{proof}
  (i)
  Let \(T=(V,F)\)~be an optimal solution for~\((G,w)\)
  that contains all edges of~\(G_\ell=(V,E_\ell)\)
  and a minimum number of edges of~\(E\setminus E_\ell\).
  In order to
  connect the \(c\)~connected components of~\(G_\ell\),
  the graph~$T$ contains \(c-1\)~edges in~\(E\setminus E_\ell\).
  These can have at most \(2c-2\)~endpoints.
  Thus, \(|W|\leq 2c-2\).

  (ii)
  Consider the graph~\(T'=(W,F')\)
  with the edge set
  \begin{align}
    F':={}&\{\{u,v\}\subseteq W\mid\{u,v\}\in F\}\cup{}
            \{\{u,v\}\subseteq W\mid \{u,v\}\in\pE\text{ and }\pw(\{u,v\})=0\}.\label{f'}
  \end{align}
  \looseness=-1
  We show that \(T'\)~is a solution to
  the \MCCS{} instance~\(\pI=(\pG{},\pw{},\col,\ell,C)\).
  We first analyze its cost.
  By \cref{def:plb},
  \(\ell(v)\leq\max_{\{u,v\}\in F}w(\{u,v\})\)
  for all \(v\in V\).
  By \cref{def:padded},
  \(\pw(\{u,v\})\leq w(\{u,v\})\)
  if \(\{u,v\}\in F\) and
  \(\pw(\{u,v\})=0\) if \(\{u,v\}\in F'\setminus F\).
  Thus,
  the cost of~\(T'\)
  as a solution to~\(\pI\) is
  \begin{align*}
    \sum_{v \in W} \max \Bigl\{ 0, \max_{\{u,v\}\in F'}\pw(\{u,v\})-\ell(v) \Bigr\}
    &\leq    \sum_{v \in V} \max \Bigl\{ 0, \max_{\{u,v\}\in F'}\pw(\{u,v\})-\ell(v) \Bigr\}\\
      \leq\sum_{v \in V} \max \Bigl\{ 0, \max_{\{u,v\}\in F}w(\{u,v\}) -\ell(v)\Bigr\}
      &\leq\sum_{v \in V} \Bigl(\max_{\{u,v\}\in F}w(\{u,v\}) -\ell(v) \Bigr)
    =\opt(I)-\smashoperator{\sum_{v \in V}}\ell(v).
  \end{align*}
  By assumption,
  \(\col\) is a bijection between~\(W\) and~\(C\).
  Thus,
  in order to show that \(T'\)~is a solution to~\(\pI\),
  it remains to show that it is
  connected.  
  Towards a contradiction,
  assume that \(T'=(W,F')\)~is not connected.
  Choose~\(u,v\in W\)
  that are disconnected in~\(T'\)
  and have minimum distance in~\(T\),
  as measured
  as the number
  of edges on a shortest \(u\)-\(v\)-path~\(p\) in~\(T\).
  By \eqref{f'},
  all edges of~\(T=(V,F)\) between vertices in~\(W\)
  are also in~\(T'\).
  Thus, \(p\)~has no inner vertices in~\(W\).
  By choice of~\(W\),
  it follows that
  all edges of~\(p\) are in~\(E_\ell\)
  and, consequently,
  \(u\) and \(v\) are
  in the same connected component of~\(G_\ell\).
  By \cref{def:padded}, \(\pw(\{u,v\})=0\).
  By \eqref{f'},
  we get \(\{u,v\}\in F'\),
  contradicting our assumption.
\end{proof}

\noindent
The next lemma is the converse to \cref{lem:mpsc-mcpu}:
we show that \cref{alg:ccs}
in line~\ref{lin:consF}
recovers a solution to~\(I\)
of cost at least \(\opt(\pI)+\sum_{v\in V}\ell(v)\)
from the \MCCS{} instance~\(\pI\) solved in
line~\ref{lin:onemcpusol}.
\begin{lemma}\label{lem:mcpu-mpsc}
  \label{rem:dp}
  Let \(I:=(G=(V,E),w)\)~be a \miPoSyCo{} instance,
  let \(\ell\colon V\to\N\)~be \plb{}s,
  and let \(\col\colon V\to C\)~be a
  coloring such that, for \(i\ne j\), the sets of colors of vertices
  in the connected components~\(G_\ell^i\) and~\(G_\ell^j\)
  of~\(G_\ell=(V,E_\ell)\) are disjoint.

  If \(T=(W,F)\)~is an optimal solution to
  the \MCCS{} instance~\(\pI=(\pG{},\pw{},\allowbreak \col,\ell,C)\),
  then \(T'=(V,(F\setminus\ppE)\cup E_\ell)\) is a solution for~\((G,w)\)
  of cost at most
  \[
    \opt(\pI)+\smashoperator{\sum_{v\in V}}\ell(v).
  \]
\end{lemma}

\begin{proof}
  As required by the definition of \MCCS{} (\cref{prob:mcpu}), \(T\)~is connected and
  contains exactly one vertex of each color in~\(C\).
  Since the color sets of distinct connected components of~\(G_\ell=(V,E_\ell)\) are disjoint,
  \(T\)~contains at least one vertex of each connected component of~\(G_\ell\).
  By construction, \(T'\) contains all edges of~\(G_\ell\)
  and all edges of~\(T\) between different connected components of~\(G_\ell\).
  Thus, \(T'\)~is connected.

  It remains to analyze the cost of~\(T'\)
  as a solution to the \miPoSyCo{} instance~\((G,w)\).
  To this end, let \(F'=(F\setminus\ppE)\cup E_\ell\)
  be the edge set of~\(T'\)
  and
  observe that, by \cref{def:oblig,def:padded},
  \begin{itemize}
  \item for all edges \(\{u,v\}\in F'\cap E_\ell\supseteq F'\setminus F\),
    one has \(\pw(\{u,v\})=0\leq w(\{u,v\})\leq\min\{\ell(u),\ell(v)\}\),
  \item for all edges \(\{u,v\}\in F'\setminus E_\ell\subseteq E\setminus \ppE\),
    one has \(\pw(\{u,v\})=w(\{u,v\})\),
  \end{itemize}
  \looseness=-1
  and that no vertex~\(v\in V\setminus W\)
  is incident to edges in~\(F\).
  Thus, the cost of~\(T'\) as solution to the \miPoSyCo{}
  instance~\((G,w)\) is
  \begin{align*}
    &\sum_{v\in V} \max_{\{u,v\}\in F'}w(\{u,v\})
    \leq\sum_{v \in V} \max \Bigl\{ \ell(v), \max_{\{u,v\}\in F'}w(\{u,v\}) \Bigr\}=\sum_{v \in V} \max \Bigl\{ \ell(v), \max_{\{u,v\}\in F'}\pw(\{u,v\}) \Bigr\}\\
    &\leq\sum_{v \in V} \max \Bigl\{ \ell(v), \max_{\{u,v\}\in F}\pw(\{u,v\}) \Bigr\}
 =\sum_{v\in V}\max\Bigl\{0,\max_{\{u,v\}\in F}\pw(\{u,v\})-\ell(v)\Bigr\}
      +\sum_{v\in V}\ell(v)\\
    &=\sum_{v\in W}\max\Bigl\{0,\max_{\{u,v\}\in F}\pw(\{u,v\})-\ell(v)\Bigr\}
      +\sum_{v\in V}\ell(v)=\opt(\pI)      +\sum_{v\in V}\ell(v).\qedhere
  \end{align*}
\end{proof}

\noindent
In \cref{lem:mcpu-mpsc,lem:mpsc-mcpu},
we have shown how to translate between
optimal solutions of the input \mpsc{} instance~\(I\)
and the optimal solutions of the \MCCS{} instance~\(\pI\)
solved in line~\ref{lin:onemcpusol},
given a suitable vertex coloring.
To prove the correctness of \cref{alg:ccs},
it remains to combine these lemmas
and analyze the probability of a suitable
coloring
in order to show that \cref{alg:ccs}
returns an optimal solution
with probability at least~\(1-\varepsilon\).

\begin{proposition}
  \cref{alg:ccs} is correct.
\end{proposition}

\begin{proof}
  By \cref{lem:mpsc-mcpu},
  there is an optimal solution \(T=(V,F)\) to the \miPoSyCo{}
  instance~\(I=(G,w)\)
  with vertex lower bounds~\(\ell\colon V\to\N\)
  such that the set \(W\subseteq V\) of vertices
  incident to an edge in~\(T\) that is not in~\(G_\ell\)
  satisfies \(|W|\leq 2c-2\),
  where \(c\)~is the number
  of connected components of~\(G_\ell\).
  Thus, in at least one iteration of the loop in line~\ref{lin:comploop},
  we will have that \(\nc_i\)~is the number of vertices of
  the connected component~\(G_\ell^i\) that are contained in~\(W\).
  If any of the colorings in line~\ref{lin:randcol}
  colors the vertices of~\(W\) in \(|W|\)~pairwise distinct colors,
  then, by \cref{lem:mpsc-mcpu},
  the \MCCS{} instance~\(\pI\) solved in
  line~\ref{lin:onemcpusol},
  and therefore the solution
  selected in line~\ref{lin:mcpusol},
  has cost at most~\(\opt(\pI)\leq \opt(I)-\sum_{v\in V}\ell(v)\).
  By \cref{lem:mcpu-mpsc},
  the \miPoSyCo{} solution returned in line~\ref{lin:consF}
  then has cost at most~\(\opt(\pI)+\sum_{v\in V}\ell(v)\leq\opt(I)\),
  implying that it is optimal.
  It remains to analyze
  the probability of a suitable coloring.

  Since line~\ref{lin:colchoice} chooses pairwise disjoint color sets
  for the components~\(G_\ell^i\),
  line~\ref{lin:randcol} colors the vertices of~\(W\)
  in \(|W|\)~pairwise distinct colors if and only if,
  for each \(i\in\{1,\dots,c\}\),
  the \(\nc_i\)~vertices of~\(W\)
  in component~\(G_\ell^i\) are colored in \(\nc_i\)~pairwise distinct colors.
  Call this event~\(A_i\) for \(i\in\{1,\dots,c\}\).
  There are \(\nc_i^{\nc_i}\)~possibilities to color the \(\nc_i\)~vertices of~\(W\)
  in component~\(G_\ell^i\).
  Out of these,
  there are \(\nc_i!\)~possibilities to color them in pairwise distinct colors.
  Thus, \(\Pr[A_i]=\nc_i!/\nc_i^{\nc_i}\) for each component~\(G_\ell^i\).
  Since \(A_i\) and~\(A_j\) are independent,
  the probability that all vertices of~\(W\) get pairwise distinct colors is
  \[
    p:=\Pr\Bigl[\bigcap_{i=1}^cA_i\Bigr]=\prod_{i=1}^c\Pr[A_i]=\prod_{i=1}^c\frac{\nc_i!}{\nc_i^{\nc_i}}.
  \]
  If \(p=1\),
  then
  the only repetition of the loop in line~\ref{lin:repeat}
  yields a suitable coloring.
  If \(p<1\),
  then the probability that none of its \(t\)~iterations
  yields a suitable coloring is
  \((1-p)^t=(1-p)^{\ln\varepsilon/\ln(1-p)}=(1-p)^{\log_{1-p}\varepsilon}=\varepsilon\).
\end{proof}

\noindent
Having shown that
\cref{alg:ccs} is correct,
to prove \cref{thm:ccs},
it remains to analyze the running time
of \cref{alg:ccs} and to derandomize it.
To analyze its running time,
we now show
how to efficiently
solve the \MCCS{} instances
in line~\ref{lin:onemcpusol} of~\cref{alg:ccs}.

\subsection{Solving \MCCS{} in line~\ref{lin:onemcpusol}
of \cref{alg:ccs}}
\label{sec:dp}
In the previous section,
we have shown that \cref{alg:ccs}
is correct.
To carry it out efficiently,
we show in this section how to
solve the \MCCS{} instances in line~\ref{lin:onemcpusol}:

\begin{proposition}
  \label[proposition]{lem:mcpu}
  The \MCCS{} instance
  in line~\ref{lin:onemcpusol} of \cref{alg:ccs}
  can be solved in time
  \[O(3^{|C|}m+ 2^{|C|}nm+nm\log n).\]
\end{proposition}

\noindent
To prove \cref{lem:mcpu},
we use a dynamic programming algorithm inspired
by an algorithm
used for finding signalling pathways
in biological networks~\citep{SIKS06}:
it finds trees containing one vertex
of each color in a vertex-colored graph.
Our case is complicated by the
non-standard goal function~\eqref{eq:mccs} of \MCCS{}
and that we do not want
to compute the graph~\(\pG{}\) explicitly:
its size might be
quadratic in that of~\(G\),
which leads to prohibitively large
running times when working on~\(\pG{}\).%
\footnote{Indeed,
  in the conference version
  of this article \citep{BBNN17},
  we worked directly on~\(\pG{}\),
  which led to a running time of \(O(3^{|C|}n^4)\)
  for carrying out line~\ref{lin:onemcpusol}.
  In experiments this turned out to be inferior
  to CPLEX or even brute force.}
We will thus have to compute
the optimal solution for~\(\pI:=(\pG,\pw,\col,\ell,C)\)
by looking at the input graph~\(G\) only.

In the following,
we will use some simplifying assumptions and conventions,
which clearly do not influence
the optimal solutions to \miPoSyCo{} and \MCCS{}.
\begin{assumption}\label[assumption]{ass}
  For each vertex~\(v\in V\),
  there is a loop~\(\{v\}\in E\) of weight \(w(\{v\})=0\).
  Consequently,
  by \cref{def:padded},
  also \(\{v\}\in\pE\) and \(\pw(\{v\})=0\).
  For any \(e\notin\pE\),
  we assume \(\pw(e)=\infty\).
\end{assumption}

\noindent
To use dynamic programming,
we formally define the subproblems
that we are going to solve
using a recurrence relation.

\begin{definition}[{$P(v,q,C')$, $\Phi(v,q,T)$, $D[v,q,C']$}]
  For any color set~\(C'\subseteq C\)
  and any edge \(\{v,q\}\in \pE\)
  (possibly, \(v=q\)),
  we denote by \(P(v,q,C')\) the subproblem
  of computing
  a feasible solution~\(T=(W,F)\)
  to the \MCCS{} instance~\((\pG,\pw,\col,\ell,C')\)
  that minimizes
  \begin{align*}
    \Phi(v,q,T):=\max\{0,\pw(\{v,q\})-\ell(v)\}+
    \smashoperator{\sum_{v' \in W\setminus\{v\}}} \max \Bigl\{ 0, \max_{\{u,v'\}\in F}\pw(\{u,v'\}) -\ell(v')\Bigr\}
  \end{align*}
  under the constraints that~\(v\in W\) and
  \begin{align}
    \max\{0,\pw(\{v,q\})-\ell(v)\}&\geq \max \Bigl\{ 0, \max_{\{u,v\}\in F}\pw(\{u,v\})-\ell(v) \Bigr\}
                                    \label{eq:mcpu-constr}
  \end{align}
  \looseness=-1
  (such a solution might not exist
  for some choices of~\(v\) and~\(q\)).
  We denote the cost of an optimal solution to~\(P(v,q,C')\)
  by
  \[
    D[v,q,C']:=\min\{\Phi(v,q,T)\mid T\text{ is a feasible solution to }P(v,q,C')\}.
  \]
\end{definition}

\noindent
Note that
the only difference between~\(\Phi(v,q,T)\)
and the goal function~\eqref{eq:mccs} of \MCCS{}
is that
the vertex \(v\)~pays exactly \(\max\{0,\pw(\{v,q\})-\ell(v)\}\).
However,
by constraint~\eqref{eq:mcpu-constr},
\(v\)~still pays at least the heaviest edge incident to~\(v\) in~\(T\).

Since we want to compute~\(D[v,q,C']\)
without looking at~\(\pG\),
the following lemma,
which exploits \cref{ass},
will be helpful.
\begin{lemma}\label{lem:oldweight}
  For each edge~\(\{v,q\}\in\pE\) %
  there is an edge~\(\{v,q'\}\in E\)
  with \(\pw(\{v,q\})=\pw(\{v,q'\})\)
  and 
  \(D[v,q,C']=D[v,q',C']\).
\end{lemma}
\begin{proof}
  If \(\{v,q\}\in E\),
  then choose \(q':=q\).
  Otherwise,
  \(\pw(\{v,q\})=0=\pw(\{v\})\) %
  by \cref{def:padded,ass}.
  Since \(\{v\}\in E\),
  one can choose \(q':=v\).
\end{proof}
\noindent
By \cref{lem:oldweight},
we can compute
the cost of an optimal solution to the \MCCS{} instance~\(\pI\)
as
\begin{equation}
\min_{\{v,q\}\in\pE}D[v,q,C]=\min_{\{v,q\}\in E}D[v,q,C].\label{eq:whereopt}
\end{equation}
For \(|C'|\leq 1\),
the value of \(D[v,q,C']\)~can be easily computed:
\begin{observation}
  \label{obs:onecolor}
  \(D[v,q,\{\col(v)\}]=\max\{0,\pw(\{v,q\})-\ell(v)\}\)
  and \(D[v,q,C']=\infty\) if \(\col(v)\notin C'\).

\end{observation}
\noindent
We now compute
\(D[v,q,C']\)
for \(|C'|\geq 2\).
To this end,
we distinguish three roles
that a vertex~\(v\)
might play in an optimal solution~\(T\)
to \(P(v,q,C')\),
which, without loss of generality,
is a tree. Vertex~\(v\) might be
\begin{itemize}
\item a cut vertex of~\(T\) (\cref{lem:innervert}), or
\item a leaf and not the only vertex of its connected
  component of~\(G_\ell\) in~\(T\) (\cref{lem:leafnotalone}), or
\item a leaf and the only vertex of its
  connected component of~\(G_\ell\) in~\(T\) (\cref{lem:leafalone}).
\end{itemize}
Herein,
note that all
recurrence relations
that we prove for~\(D[v,q,C']\)
will not refer to~\(\pG{}\),
but to the input graph~\(G\)
only,
since we want to avoid the expensive construction of~\(\pG{}\).

\begin{lemma}\label[lemma]{lem:innervert}
  For any \(\{v,q\}\in E\)
  and \(C'\subseteq C\) with \(|C'|\geq 2\),
  \begin{equation}
    D[v,q,C']\leq D_1[v,q,C']:=\min
    \begin{Bmatrix*}[l]
      D[v,q,C_1]+D[v,q,C_2]-\max\{0,\pw(\{v,q\})-\ell(v)\}\\
      \text{\quad for all } C_1\subsetneq C'\text{ and }C_2\subsetneq C'\\
      \text{\qquad such that }C_1\cup C_2=C'
      \text{ and }C_1\cap C_2=\{\col(v)\}
    \end{Bmatrix*}.
    \label{eq:dp1}
  \end{equation}
  If there is an optimal solution~\(T\)
  to~$P(v,q,C')$
  with cut vertex~$v$,
  then~$D[v,q,C'] = D_1[v,q,C']$.
\end{lemma}

\begin{proof}
  (\(\leq\))
  By taking the unions of the vertex sets
  and edge sets
  of two optimal solutions
  for \(P(v,q,C_1)\) and \(P(v,q,C_2)\)
  with \(C_1\cup C_2=C'\) and \(C_1\cap C_2=\{\col(v)\}\),
  we get a feasible solution~\(T'\)
  for \(P(v,q,C')\).
  Since
  these are edge-disjoint and intersect only in~\(v\),
  we get
  \begin{align*}
    D[v,q,C']\leq \Phi(v,q,T')&=D[v,q,C_1]+D[v,q,C_2]-\max\{0,\pw(\{v,q\})-\ell(v)\}.
  \end{align*}

    (\(\geq\))
  Solution~\(T\) decomposes into
  two proper subgraphs~\(T_1\) and~\(T_2\)
  only intersecting in~\(v\).
  For~\(i\in\{1,2\}\),
  let \(C_i\)~be the set of colors of the vertices of~\(T_i\).
  Then, one has
  \(C_1\subsetneq C'\), \(C_2\subsetneq C'\),
  \(C_1\cup C_2=C'\), and
  \(C_1\cap C_2=\{\col(v)\}\).
  For \(i\in\{1,2\}\),
  \(T_i\) is a feasible solution to
  \(P(v,q,C_i)\).
  Thus,
  \begin{align*}
    D[v,q,C']= \Phi(v,q,T)&=\Phi(v,q,T_1)+\Phi(v,q,T_2)-\max\{0,\pw(\{v,q\})-\ell(v)\}\\
                          &\geq D(v,q,C_1)+D[v,q,C_2]-\max\{0,\pw(\{v,q\})-\ell(v)\}.\qedhere
  \end{align*}
\end{proof}

\begin{lemma}
  \label[lemma]{lem:leafnotalone}
  For any \(\{v,q\}\in E\)
  and \(C'\subseteq C\) with \(|C'|\geq 2\)
  such that the connected component of~\(G_\ell\)
  containing~\(v\)
  contains a vertex
  with color in~\(C'\setminus\{\col(v)\}\),
  \begin{equation}
    D[v,q,C']\leq D_2[v,q,C']:=\min_{\{u,q'\}\in E}D[u,q',C'\setminus\{\col(v)\}]+\max\{0,\pw(\{v,q\})-\ell(v)\}.\label{eq:dp2}
  \end{equation}
  If there is an optimal solution~\(T\)
  to~$P(v,q,C')$
  with a leaf~$v$,
  then~$D[v,q,C'] = D_2[v,q,C']$.
\end{lemma}

\begin{proof}
  (\(\geq\))
  \(T'=T-\{v\}\)
  is a feasible solution
  with cost
  at least~\(D[u,q',C'\setminus\{\col(v)\}]\)
  for~\(P(u,q',C'\setminus\{\col(v)\})\),
  where \(u\)~is any vertex in~\(T'\)
  and \(\{u,q'\}\in\pE\)~is an edge
  incident to~\(u\) in~\(T'\)
  that maximizes~\(\pw(\{u,q'\})\).
  Thus, using \cref{lem:oldweight}, we get
  \begin{align*}
    D[v,q,C']&= \Phi(u,q',T')+\max\{0,\pw(\{v,q\})-\ell(v)\}\\
             &\geq D[u,q',C'\setminus\{\col(v)\}]+\max\{0,\pw(\{v,q\})-\ell(v)\}\\
    &\geq \min_{\{u',q''\}\in E}D[u',q'',C'\setminus\{\col(v)\}]+\max\{0,\pw(\{v,q\})-\ell(v)\}.
  \end{align*}
  (\(\leq\))
  Let \(\{u,q'\}\in E\)~be
  an edge minimizing~\(D[u,q',C'\setminus\{\col(v)\}]\),
  let \(T'\)~be an optimal solution
  to $P(u,q',C\setminus\{\col(u)\})$,
  and \(u^*\)~be a vertex in~\(T'\)
  such that \(\col(u^*)\in C'\setminus\{\col(v)\}\)
  and \(u^*\)~is in the same connected component of~\(G_\ell\)
  as~\(v\).
  By \cref{def:padded},
  \(\{u^*,v\}\in \pE\) and \(\pw(\{u^*,v\})=0\).
  Thus,
  adding this edge to~\(T'\), we get
  a connected subgraph~\(T^*=(W,F)\) of~\(\pG\)
  containing all colors of~\(C'\).
  It is a feasible solution for~\(P(v,q,C')\)
  since \(T^*\)~satisfies constraint~\eqref{eq:mcpu-constr}, that is,
  \begin{align*}
    \max\{0,\pw(\{v,q\})-\ell(v)\}&\geq 0= \max\{0,\pw(\{u^*,v\})-\ell(v)\}= \max \Bigl\{ 0, \max_{\{u,v\}\in F}\pw(\{u,v\})-\ell(v) \Bigr\},
  \end{align*}
  where the last equality is due to
  the fact that \(\{u^*,v\}\)~is the
  only edge incident to~\(v\) in~\(T^*\).
  Thus,
  \begin{align*}
    D[v,q,C']\leq\Phi(v,q,T^*)&=\Phi(u,q',T')+\max\{0,\pw(\{v,q\})-\ell(v)\}\\
                              &=\min_{\{u',q''\}\in E}D[u',q'',C'\setminus\{\col(v)\}]
                                +\max\{0,\pw(\{v,q\})-\ell(v)\}.\qedhere
  \end{align*}
\end{proof}
\begin{lemma}
  \label[lemma]{lem:leafalone}
  For any \(\{v,q\}\in E\)
  and \(C'\subseteq C\) with \(|C'|\geq 2\)
  such that the connected component of~\(G_\ell\)
  containing~\(v\)
  does \emph{not} contain a vertex
  with color in~\(C'\setminus\{\col(v)\}\),
  \begin{equation}
    D[v,q,C']\leq D_3[v,q,C']:=\min
    \begin{Bmatrix*}[l]
      D[u,q',C'\setminus\{\col(v)\}]+\max\{0,\pw(\{v,q\})-\ell(v)\}\\
      \text{\quad for all }u\in N_G(v)\text{ and }q'\in N_G(u)\\
      \text{\qquad such that }\pw(\{u,q'\})\geq \pw(\{u,v\})\\
      \text{\qquad and }\pw(\{v,q\})\geq \pw(\{u,v\})
    \end{Bmatrix*}.
    \label{eq:dp3}
  \end{equation}
  If there is an optimal solution~\(T\)
  to~$P(v,q,C')$
  with a leaf~$v$,
  then~$D[v,q,C'] = D_3[v,q,C']$.
\end{lemma}

\begin{proof}{}
  $(\leq)$
  Let \(T'\)~be an optimal solution
  to \(P(u,q',C'\setminus\{\col(v)\})\)
  for any \(u\in N_G(v)\)
  and \(q'\in N_G(u)\)
  such that
  \(\pw(\{u,q'\})\geq \pw(\{u,v\})\)
  and
  \(\pw(\{v,q\})\geq \pw(\{u,v\})\).
  Adding the edge~\(\{u,v\}\) to~\(T'\)
  gives a feasible solution~\(T'\) for \(P(v,q,C')\).
  Thus
  \begin{align*}
    D[v,q,C']\leq \Phi(v,q,T')&=D[u,q',C'\setminus\{\col(v)\}]+\max\{0,\pw(\{v,q\})-\ell(v)\}.
  \end{align*}
  \((\geq)\)
  Let \(u\)~be the neighbor of~\(v\) in~\(T\).
  Since \(T\)~contains no
  other vertices than~\(v\)
  from the connected component of~\(G_\ell\),
  by \cref{def:padded}
  it follows that \(\{u,v\}\in E\),
  or, equivalently,
  \(u\in N_G(v)\).
  Now, let
  \(\{u,q'\}\) be an edge in~\(T\)
  maximizing~\(\pw(\{u,q'\})\).
  Then
  \(T-\{v\}\)
  is a feasible solution for \(P(u,q',C'\setminus\{\col(v)\})\) and
  \begin{align*}
    D[v,q,C']=\Phi(v,q,T)&=\Phi(u,q',T-\{v\})+\max\{0,\pw(\{v,q\})-\ell(v)\}\\
    &\geq
      D[u,q',C'\setminus\{\col(v)\}]+\max\{0,\pw(\{v,q\})-\ell(v)\}\\
    &=D[u,q^*,C'\setminus\{\col(v)\}]+\max\{0,\pw(\{v,q\})-\ell(v)\}
  \end{align*}
  for some \(\{u,q^*\}\in E\)
  with \(\pw(\{u,q^*\})=\pw(\{u,q'\})\)
  by \cref{lem:oldweight}.
  Note that \(q^*\in N_G(u)\)
  with
  \(\pw(\{u,q^*\})=\pw(\{u,q'\})\geq\pw(\{u,v\})\)
  since \(\{u,q'\}\) maximizes
  \(\pw(\{u,q'\})\)  among the edges incident to~\(u\)
  in~\(T\).
\end{proof}

\noindent
This completes
our recurrence relations
for computing~\(D[v,q,C']\),
which we will now use to prove \cref{lem:mcpu}.

\begin{proof}[Proof of \cref{lem:mcpu}.]
  We first compute all connected components
  of~\(G\) in \(O(n+m)\)~time
  using depth\hyp first search,
  marking each vertex
  with the number of the connected
  component it belongs to.
  This allows us to access~\(\pw(\{v,q\})\)
  for any edge~\(\{v,q\}\in E\)
  in constant time:
  if \(v\)~and \(q\) are in one component,
  then \(\pw(\{v,q\})=0\) by \cref{def:padded}.
  Otherwise, \(\pw(\{v,q\})=w(\{v,q\})\).
  
  We now sort the neighbors~\(q\)
  of each vertex~\(v\)
  by non\hyp increasing edge weight~\(\pw(\{v,q\})\)
  in \(O(n\log n)\)~total time.
  Then,
  in order to compute~\(D_3[v,q,C']\) in \eqref{eq:dp3}
  quickly later,
  for each edge~\(\{v,q\}\in E\)
  and \(u\in N_G[v]\) with
  \(\pw(\{v,q\})\geq \pw(\{u,v\})\),
  we precompute
  \[
    X[v,q,u]:=\max\{j\in\{1,\dots,\deg(u)\}\mid
    \pw(\{u,q_j\})\geq \pw(\{u,v\})\},
  \]
  where \(N_G(u)=\{q_1,\dots,q_{\deg(u)}\}\)
  such that \(\pw(\{u,q_i\})\geq \pw(\{u,q_{i+1}\})\)
  for all~\(i\in\{1,\dots,j-1\}\).
  The computation of \(X[v,q,u]\)
  for all \(\{v,q\}\in E\) and \(u\in N_G(v)\)
  works in \(O(mn\log n)\)~total time
  using binary search on~\(N_G(u)\).

  Then, we compute
  the table entries~\(D[v,q,C']\)
  for all~\(\{v,q\}\in E\) and~\(C'\subseteq C\).
  The optimum of the \MCCS{} instance
  is then given by \eqref{eq:whereopt}.
  We compute the table entries
  for increasing cardinality of~\(C'\).
  While computing
  \(D[v,q,C']\)
  for all~\(\{v,q\}\in E\)
  and fixed~\(C'\subseteq C\),
  we also precompute
  \[
    Y[C']:=\min\{D[v,q,C']\mid \{v,q\}\in E\}
  \]
  increasing the running time
  by only a constant factor.
  Moreover,
  while computing
  \(D[v,q,C']\)
  for fixed~\(C'\subseteq C\),
  fixed~\(v\in V\),
  and all~\(q\in N_G(v)\)
  by non\hyp increasing edge weight,
  we precompute
  \[
    Z[v,j,C']:=\min\{D[v,q_i,C'\setminus\{\col(v)\}]\mid
    i\in\{1,\dots,j\}\},
  \]
  where \(N_G(v)=\{q_1,\dots,q_{\deg(v)}\}\)
  such that \(\pw(\{v,q_i\})\geq \pw(\{v,q_{i+1}\})\)
  for all~\(i\in\{1,\dots,j-1\}\).
  This increases the running time
  only by a constant factor since
  \(Z[v,j,C']\)~is
  computable in constant time
  from~\(Z[v,j-1,C']\)
  and~\(D[v,q_j,C'\setminus\{\col(v)\}]\).
  
  We now describe how to
  compute the table entries~\(D[v,q,C']\).
  By \cref{obs:onecolor},
  the \(O(m)\)~table entries~\(D[v,q,C']\)
  for \(\{v,q\}\in E\) and~\(|C'|=1\)
  can be computed in constant time each.

  For \(|C'| \geq 2\),
  we compute \(D[v,q,C']\)
  under the assumption
  that all table entries for~\(C''\subsetneq C'\)
  are already known.
  Vertex~\(v\) is either a cut vertex
  or a leaf of an optimal solution~\(T\)
  to \(P(v,q,C')\),
  which, without loss of generality,
  is a tree.
  Thus,
  if there is a vertex~\(u\)
  in the same component of~\(G_\ell\)
  as \(v\) and \(\col(u)\in C'\),
  then we compute
  \(D[v,q,C']=\min\{D_1[v,q,C'],D_2[v,q,C']\}\)
  using \eqref{eq:dp1} and \eqref{eq:dp2}.
  Otherwise, we compute
  \(D[v,q,C']=\min\{D_1[v,q,C'],D_3[v,q,C']\}\)
  using \eqref{eq:dp1} and \eqref{eq:dp3}.

  One can compute \(D_1[v,q,C']\) from \eqref{eq:dp1}
  in \(O(3^{|C|}m)\)~total time
  for all~\(\{v,q\}\in E\)
  and~\(C'\subseteq C\)
  since,
  in total,
  there are at most
  \(3^{|C|}\)~ways to choose~\(C_1\cup C_2=C'\) for all~\(C'\subseteq C\)
  such that
  \(C_1\cap C_2=\{\col(v)\}\) for some~\(v\in V\):
  each element except \(\col(v)\)
  is either in~\(C_1\), in~\(C_2\), or in \(C\setminus(C_1\cup C_2)\).

  One can compute \(D_2[v,q,C']\) from \eqref{eq:dp2}
  in constant time for each~\(\{v,q\}\in E\)
  and~\(C'\subseteq C\)
  using~\(Y[C'\setminus\{v\}]\).
  Thus,
  the total time spent
  computing the \(D_2[v,q,C']\)
  is \(O(2^{|C|}\cdot m)\).

  Finally,
  \(D_3[v,q,C']\) from \eqref{eq:dp3}
  can be computed
  in \(O(n)\)~time
  for each~\(\{v,q\}\in E\) and~\(C'\subseteq C\):
  we iterate over each \(u\in N_G(v)\)
  with \(\pw(\{v,q\})\geq \pw(\{u,v\})\)
  and, for each of them,
  look up~\(Z[u,j,C'\setminus\{v\}]\)
  for \(j=X[v,q,u]\).
  It follows
  that the total time
  to compute \(D_3[v,q,C']\) from~\eqref{eq:dp3}
  is \(O(2^{|C|}mn)\).
\end{proof}
\begin{remark}
\cref{lem:mcpu} directly yields \cref{prop:exalg}:
for solving an \miPoSyCo{} instance~\((G,w)\) with~\(G=(V,E)\),
we can simply choose \(\ell\colon V\to\N,v\mapsto 0\),
the color set~\(C=\{1,\dots,n\}\),
an arbitrary bijection \(\col\colon V\to C\).
Then,
\(\pG=G\) and \(\pw=w\)
and we solve the \MCCS{} instance~\((\pG,\pw,\col,\ell,C)\)
in \(O(3^n\cdot m)\)~time
by \cref{lem:mcpu},
which is equivalent
to~\((G,w)\) by \cref{lem:mcpu-mpsc,lem:mpsc-mcpu}.
\end{remark}

\subsection{Running time and derandomization of \cref{alg:ccs}}
\label{sec:rt}
We can finally prove
the running time, error probability,
and show the derandomization of \cref{alg:ccs},
thus proving \cref{thm:ccs}.

\begin{proof}[Proof of \cref{thm:ccs}.]
  To analyze the running time of \cref{alg:ccs},
  note that
  there are \({2c-2\choose c}\)~possibilities
  to enumerate all~\(\nc_1,\dots,\nc_c\in\N^+\) with \(\sum_{i=1}^c \nc_i\leq 2c-2\)
  in line~\ref{lin:comploop}.
  Using Stirling's approximation
  \[
    \sqrt{2\pi n}\cdot\left(\frac{n}{e}\right)^n\leq n!\leq e^{1/12n}\sqrt{2\pi n}\cdot\left(\frac{n}{e}\right)^n,
  \]
  we get that the number of iterations of the loop in line~\ref{lin:comploop} is
  
  \begin{align*}
    {2c-2\choose c}&=\frac{(2c-2)!}{c!\cdot (c-2)!}\in O\left(\frac{\sqrt{2c-2}}{\sqrt{c}\cdot\sqrt{c-2}}\cdot\frac{(2c-2)^{2c-2}}{e^{2c-2}}\cdot \frac{e^c\cdot e^{c-2}}{c^{c}\cdot (c-2)^{c-2}}\right)
                     =O\left(\frac{(2c-2)^{2c-2}}{\sqrt{c}\cdot c^c\cdot (c-2)^{c-2}}\right)\\
    &=O\left(\frac{4^{c}}{\sqrt{c}}\cdot \Bigl(\frac{c-1}{c}\Bigr)^c\cdot\Bigl(\frac{c-1}{c-2}\Bigr)^{c-2}\right)=O\left(\frac{4^{c}}{\sqrt{c}}\cdot\Bigl(1-\frac{1}{c}\Bigr)^{c}\cdot\Bigl(1+\frac{1}{c-2}\Bigr)^{c-2}\right)
      =O\left(\frac{4^c}{\sqrt{c}}\right).
  \end{align*} 
  Solving
  the \MCCS{} instance with at most \(2c-2\)~colors
  in line~\ref{lin:onemcpusol}
  works in \(O(3^{2c-2}\cdot m+2^{2c-2}nm+nm\log n)\)~time
  by \cref{lem:mcpu}.
  To analyze the number~\(t\) of repetitions in line~\ref{lin:repeat},
  we use \(x\geq\ln(1+x)\) and
  again Stirling's approximation to get
  \begin{align*}
    t-1&\leq
    \ln1/\varepsilon\cdot\prod_{i=1}^c\frac{\nc_i^{\nc_i}}{\nc_i!}\leq
       \ln1/\varepsilon\cdot\prod_{i=1}^c\frac{e^{\nc_i}}{\sqrt{2\pi \nc_i}}
    \leq
      \ln1/\varepsilon\cdot \frac{e^{2c-2}}{\sqrt{2\pi}^c}.
  \end{align*}
  The running time of the algorithm is thus
  \(\ln 1/\varepsilon \cdot (4e^2/\sqrt{2\pi})^c \cdot 1/\sqrt{c}
  \cdot O(9^cm+4^cnm+nm\log n)\).

  (ii) To derandomize the algorithm,
  we use \emph{\((d,k)\)-perfect hash families}~\(\mathcal F\)
  of functions~\(f\colon\{1,\dots,d\}\to\{1,\dots,k\}\)
  such that, for each subset~\(W\subseteq\{1,\dots,d\}\)
  of size~\(k\),
  at least one of the functions in~\(\mathcal F\)
  is a bijection between~\(W\) and~\(\{1,\dots,k\}\).
  Let \(n_i\)~be the number of vertices in component~\(G_\ell^i\).
  Instead of coloring the vertices in each component~\(G_\ell^i\)
  for \(i\in \{1,\dots,c\}\) with colors from~\(C_i\)
  randomly in line~\ref{lin:randcol},
  we color them
  using all of the functions in~\(\mathcal F_i\)
  of an  \((n_i,\nc_i)\)-perfect hash family~\(\mathcal F_i\).

  One can construct an \((n_i,\nc_i)\)-perfect hash family~\(\mathcal F_i\)
  with \(e^{\nc_i}{\nc_i}^{O(\log \nc_i)}\log n_i\)~functions
  in
  \(e^{\nc_i}{\nc_i}^{O(\log \nc_i)}\cdot n_i\log n_i\)~time \citep[Theorem~5.18]{CFK+15}.
  Thus,
  in each iteration of the loop in line~\ref{lin:comploop},
  we generate the families~\(\mathcal F_i\) for all \(i\in\{1,\dots,c\}\)
  in
  \[
    \sum_{i=1}^ce^{\nc_i}{\nc_i}^{O(\log \nc_i)}n_i\log n_i\subseteq e^cc^{O(\log c)}\log n \sum_{i=1}^cn_i=e^cc^{O(\log c)}n\log n %
  \]
  time. Then, in line~\ref{lin:randcol}, we color
  the vertices of all components of~\(G_\ell\)
  according to
  \[
    \prod_{i=1}^ce^{\nc_i}{\nc_i}^{O(\log \nc_i)}\log n_i\subseteq (\log n)^ce^{2c}c^{O(c\log c)} %
  \]
  functions. Thus, the overall running time of the deterministic algorithm is
  \(c^{O(c\log c)}\cdot (4e)^{2c}\cdot(\log n)^c\cdot O(9^cm+4^cnm+nm\log n)\),
  which is \(c^{O(c\log c)}\cdot n^{O(1)}\) \citep[Exercise~3.18]{CFK+15}.
\end{proof}

\subsection{Hardness of provably effective data reduction}
\label{sec:dr}

In the previous sections,
we have seen that
the number~$c$ of connected components
of the obligatory subgraph
can be effectively used to obtain
fixed\hyp parameter algorithms.
We will experimentally support these findings
in \cref{sec:ccs-exp},
where we will also find data reduction
to play an important role.

However,
the following theorem shows that
\mpsc{} has no problem kernel of size polynomial in~$c$
unless the polynomial\hyp time hierarchy collapses
and, as is also conjectured, does not even have Turing kernels of size polynomial in~$c$ \citep{HKS+15b}.

\begin{theorem}\label[theorem]{thm:wkhard}
  \miPoSyCo{} is WK[1]-hard parameterized by the number~$c$ of
  connected components of the obligatory subgraph~$G_\ell$
  even for the trivial \plb{}s~$\ell(v)$
  equal to the least weight of any edge incident to~$v$.
  This holds even in graphs of edge weights one and two
  or in complete graphs with metric edge weights.
\end{theorem}

\begin{proof}
\citet{HKS+15b} have shown that
the problem of
checking whether a \MSC{} instance~$(U,\mathcal F)$
has
a solution of cost at most~$k$
is WK[1]-hard parameterized by~$|U|$.

As shown in \cref{lem:param-reduction},
\cref{trans:phi,trans:metric}
correctly reduce this problem to \kmiPoSyCo{}.
It is now enough to observe
that the obligatory subgraphs
of the graphs~$G=(\{s\}\uplus U\uplus\mathcal F,E)$ generated by
\cref{trans:phi,trans:metric}
have $|U|+1$ connected components:
one component consists of the vertices~$v\in \{s\}\uplus\mathcal F$, each of which has~$\ell(v)=1$,
and $|U|$ components are formed by the vertices in~$u\in U$,
each of which has~$\ell(u)=2$.
Thus,
\cref{trans:phi,trans:metric}
are polynomial parameter transformations
of \MSC{} parameterized by~$|U|$
to \kmiPoSyCo{} parameterized by~$c$.
\end{proof}

\noindent
In contrast to \cref{thm:wkhard},
we point out that,
using an approach of \citet{BFTxx},
given any~$\varepsilon>0$,
one can reduce any instance~$I$ of \mpsc{} with metric
edge weights
to an instance~$I'$ of an annotated version of the problem
with $O(c/\varepsilon)$~vertices
such that any $\alpha$-approximation for~$I'$
can be transformed to an~$(1+\varepsilon)\alpha$-approximation for~$I$  \citep{Smi20}: the annotations merely
keep track of \plb{}s and
which vertices in the remaining instance
were connected in the input instance.

\section{Experimental evaluation}
\label{sec:ccs-exp}
In this section,
we experimentally evaluate
\cref{alg:ccs}
and compare it to state
of the art ILP models due to \citet{MG05}
solved by CPLEX.
In \cref{sec:heur},
we describe two data reduction rules
to speed up the algorithms.
In \cref{sec:data},
we describe the data we tested the algorithms on.
In \cref{sec:setup},
we describe our %
test environment
and implementation.
Finally, in \cref{sec:results},
we present our experimental results.

\subsection{Data reduction}
\label{sec:heur}

The following data reduction rules
preserve the possibility to find optimal solutions.
However,
by \cref{thm:wkhard},
they cannot provably
lead to a problem kernel for \mpsc{}
with size polynomial in the number~$c$
of the connected components of the obligatory subgraph.

\paragraph{Heavy edge deletion.}
The first preprocessing step is inspired by \citet{MG05}.
The weight~\(M\) of a minimum spanning tree
is at most twice the cost of an optimal
solution to an \mpsc{} instance $(G,w)$ on a graph~$G=(V,E)$ \citep{KKKP00},
so no edge~\(e\in E\) with weight~\(w(e)>2M\)
can be part of an optimal solution.
\citet{MG05} thus suggest to delete all edges~\(e\)
with \(w(e)>2M\) from~$E$.
We take this thought further,
incorporating \plb{}s~$\ell\colon E\to\N$
and taking a possibly tighter upper bound than~$2M$.

Let $M'\leq 2M$~be the cost of a minimum spanning tree
when viewed as a solution to \mpsc{}.
By \cref{def:plb},
there is an optimal solution~$T$
in which each vertex~$v$ pays at least~$\ell(v)$.
Since the cost of~$T$ is at most~$M'$,
it cannot contain any edge~$\{u,v\}\in E$ satisfying
\[
\smashoperator{\sum_{x\in V\setminus\{u,v\}}}\ell(x)+ \max\{\ell(u), w(\{u,v\})\} + \max\{\ell(v), w(\{u,v\})\} > M'.
\]
We thus delete all such edges.

\paragraph{Redundant vertex deletion.}
While the previous data reduction rule
is applicable to any solution algorithm for \mpsc{},
the next data reduction exploits
properties of \cref{alg:ccs}.
In line~\ref{lin:mcpusol},
the connected components
of the obligatory subgraphs~$G_\ell$
are treated as cliques:
the \MCCS{} subproblem is solved on~$\pG$.
Thus,
there is always an optimal solution to the \MCCS{} subproblem
that does not contain vertices of~$\pG$
having neighbors only in their own connected component of~$G_\ell$.
We thus delete such vertices.

\subsection{Data generation}
\label{sec:data}

Due to the lack of openly available
benchmark instances for \mpsc{},
experimental works on \mpsc{}
evaluate their algorithms
on random points on a plane \citep{MG05,ACM+06,EMP17,PEM19}.
As sketched in \cref{ex:broken-grids},
in this case
the number of connected components
in the obligatory subgraph is likely
to be~\(\Theta(n)\).
Such test instances seem artificial in several application scenarios.
Moreover,
they lack any input structure,
whereas the aim of our algorithm is efficiently solving \mpsc{}
by making explicit use of input structure.
The latter presumably occurs in real-world monitoring systems
as,
in order to guarantee a long lifetime of the sensor network,
the sensor layout takes energy saving aspects into account.
This aspect was taken into consideration for the two instance types
that we describe 
in \cref{sec:faulty,sec:lakes}.

In all generated instances,
we choose as \plb{}~$\ell(v)$ %
the least weight of any edge incident to the vertex~$v$.
For vertices
incident to a single edge~$\{u,v\}$,
we set the \plb{}s $\ell(u)$ and $\ell(v)$
to cover at least the weight of~$\{u,v\}$.

\begin{figure}
    \centering
    \small
    \begin{subfigure}{0.45\textwidth}
      \raggedleft
    \begin{tikzpicture}[scale=1]
      \begin{axis}[ticks=none,
        width=9cm, height=9cm,
                     xmin=-1, xmax=10.5, ymin=-1, ymax=11,
                     mark size=2, black,
                     scatter/classes={
                         a={mark=*, fill=white},
                         b={mark=square*},
                         c={mark=*}
                     }]
            \addplot[scatter, only marks, scatter src=explicit symbolic]
                table [x expr=\thisrow{y}, y expr=10-\thisrow{x}, meta=c] {simple_example.txt};
        \end{axis}
    \end{tikzpicture}
    \label{fig:faultya}
    \end{subfigure}
    \hfill
    \begin{subfigure}{0.45\textwidth}
      \raggedleft
    \begin{tikzpicture}[scale=1]
      \begin{axis}[ticks=none,
        width=9cm, height=9cm,
                     xmin=-1, xmax=10.5, ymin=-1, ymax=11,
                     mark size=2, black,
                     scatter/classes={
                         a={mark=*, fill=white},
                         b={mark=square*},
                         c={mark=*}
                     }]
            \input{simple_example_graph.tex}
            \addplot[scatter, only marks, scatter src=explicit symbolic]
                table [x expr=\thisrow{y}, y expr=10-\thisrow{x}, meta=c] {simple_example.txt};
        \end{axis}
    \end{tikzpicture}
    \label{fig:faultyb}
    \end{subfigure}

    \caption{An instance from the ``faulty grid'' data set
    for $N=10$ and $c=3$ on the left and its optimal solution on the right.
      Sensors from distinct obligatory components are drawn using distinct marks.}
    \label{fig:faulty-grid}
\end{figure}

\subsubsection{The ``faulty grid'' data set.}
\label{sec:faulty}
This instance set
simulates the sensor fault scenario described in
\cref{ex:broken-grids}.
Grid-like sensor arrangements
minimize sensing area overlap when monitoring areas,
which minimizes the energy for sensing
and thus is important for a long lifetime of the
(usually battery-powered) sensor network~\citep{ZH05,ZEAC09}.
Thus,
in this data set,
sensors are laid out on a triangular grid,
yet we assume that several sensors fail.
The goal
is to restore the connectivity of the network
at minimum additional cost,
again in order to minimize energy consumption.
An example is shown in \cref{fig:faulty-grid}.

More specifically,
the generation of the ``faulty grid'' instance is
governed by two parameters:
the grid size~$N$ and the number of obligatory components~$c$.
An instance is generated by
assuming an infinite grid of
equilateral triangles with unit edge lengths in the plane
and taking the nodes of the grid that fall
into a $[0, N]\times[0, N]$ rectangle.
These nodes are the wireless sensors
that form the vertices of a complete graph~\(G\).
An edge $\{v, u\}$ in~\(G\)
has weight equal to the squared Euclidean
distance between~\(v\) and~\(u\),
since by the inverse-square law,
signal intensity depends inversely proportionally
on the squared distance.
Now, we select uniformly at random a fraction of vertices from~\(G\) and delete them to simulate defective sensors.
We chose the fraction to be $0.1 + \frac{1}{\sqrt{N}}$,
since it yields a small yet greater than one number of obligatory components.
After deleting vertices,
we select
graphs whose obligatory subgraph
has~$c$~connected components.

We generated instances with
$N \in \{10, 20, \dots, 80\}$,
and
$c \in \{3, 4, 5\}$.
As seen in \cref{fig:faulty-grid},
these instances
usually have one giant connected component
and several small connected components.
\setlength{\fboxsep}{0pt}
\begin{figure}
    \begin{subfigure}[t]{0.45\textwidth}
      \fbox{\includegraphics[width=\textwidth]{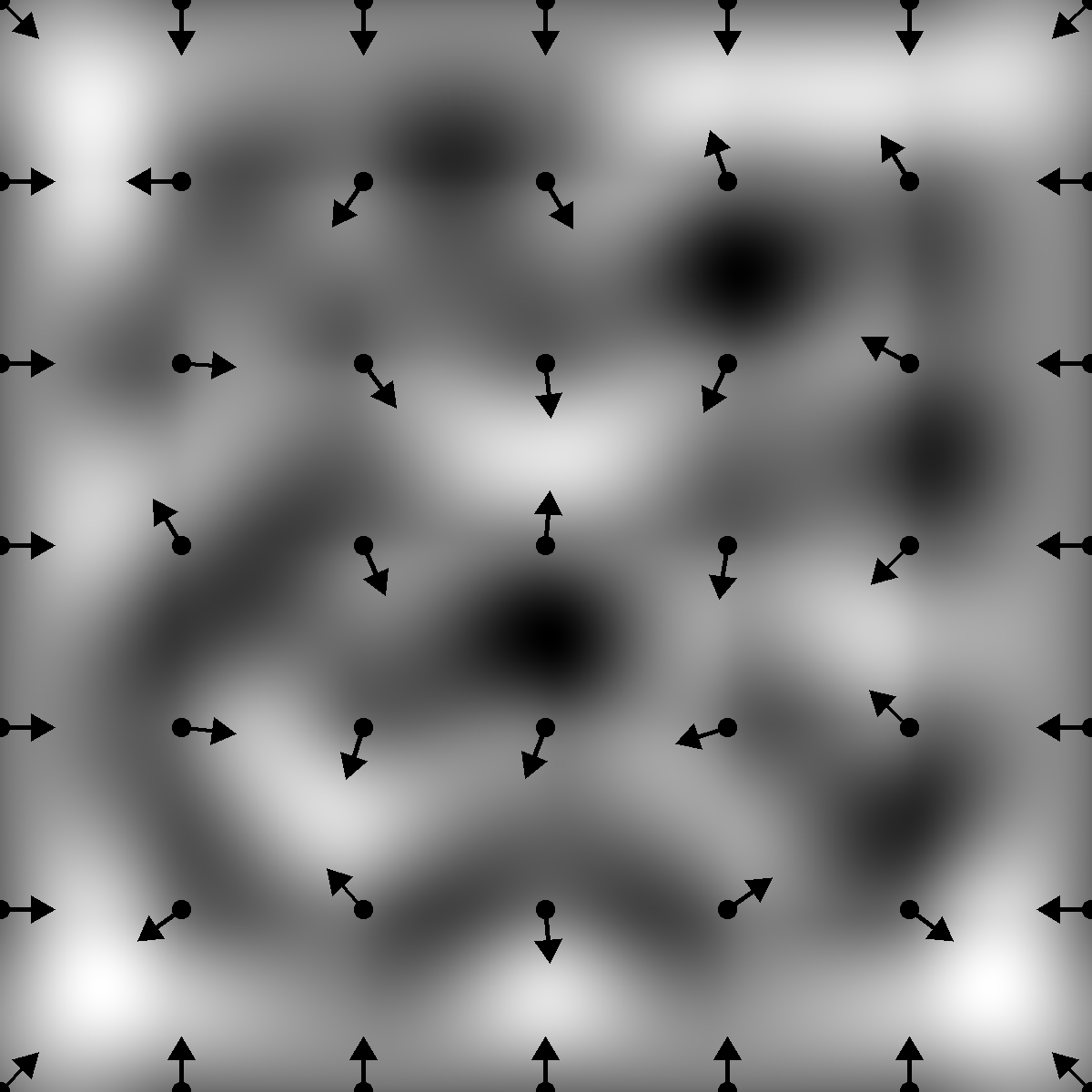}}
        \caption{
            Height map over a $[0,M-1]^2$ square
            generated from Perlin noise using randomly chosen gradient vectors
            in the square's integer points.
            A lighter shade of gray means higher.
        }
        \label{fig:lakesa}
    \end{subfigure}
    \hfill
    \begin{subfigure}[t]{0.45\textwidth}
        \fbox{\includegraphics[width=\textwidth]{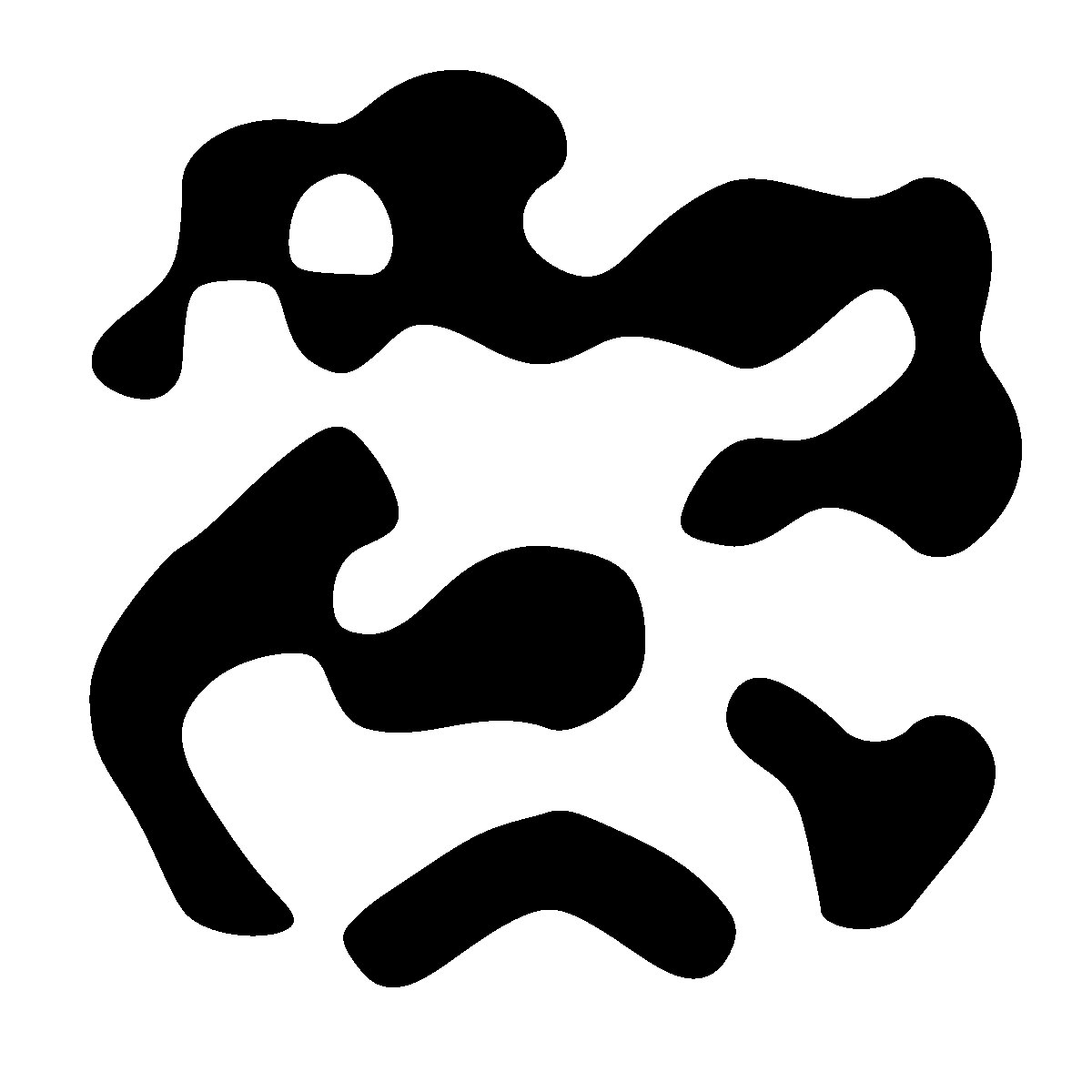}}
        \caption{
            Black areas are below the ``water level'' and form lakes.
        }
        \label{fig:lakesb}
    \end{subfigure}

    \medskip
    \begin{subfigure}[t]{0.45\textwidth}
    \begin{tikzpicture}[scale=1]
      \begin{axis}[ticks=none,
                    width=9cm, height=9cm,
                     xmin=0, xmax=16, ymin=0, ymax=16,
                     mark size=2, black,
                     scatter/classes={
                         a={mark=square*},
                         b={mark=*, fill=white},
                         c={mark=*},
                         d={mark=triangle*}
                     }]
            \addplot[scatter, only marks, scatter src=explicit symbolic]
                table [x expr=\thisrow{y}, y expr=16-\thisrow{x}, meta=c] {perlin_example.txt};
        \end{axis}
    \end{tikzpicture}
        \caption{A triangular sensor grid covers the four lakes. Distinct obligatory components are represented by distinct marks.}
        \label{fig:lakesc}
    \end{subfigure}
    \hfill
    \begin{subfigure}[t]{0.45\textwidth}
    \begin{tikzpicture}[scale=1]
      \begin{axis}[ticks=none,
                    width=9cm, height=9cm,
                     xmin=0, xmax=16, ymin=0, ymax=16,
                     mark size=2, black,
                     scatter/classes={
                         a={mark=square*},
                         b={mark=*, fill=white},
                         c={mark=*},
                         d={mark=triangle*}
                     }]
            \input{perlin_example_graph.tex}
            \addplot[scatter, only marks, scatter src=explicit symbolic]
                table [x expr=\thisrow{y}, y expr=16-\thisrow{x}, meta=c] {perlin_example.txt};
        \end{axis}
    \end{tikzpicture}
    \caption{An optimal solution.}
        \label{fig:lakesd}
    \end{subfigure}
  \caption{Three steps of generating a ``lakes'' instance for $M=7$,
    $N=16$, $c=4$, and an optimal solution.}
  \label{fig:lakes}
\end{figure}

\subsubsection{The ``lakes'' data set.}
\label{sec:lakes}
This data set corresponds to the scenario
where we have to connect components of a wireless network
that are already connected by vertex lower bounds
(or due to constraints arising in the application).
Intuitively,
one can imagine the task of measuring pollution in lakes
using triangular sensor grids.
An instance of this type is shown in
\cref{fig:lakesc}.
The lakes are generated as follows:
we generate a triangular grid and terrain data,
assign a global water level to the terrain,
and take the sensors lying in the lakes.

In more detail
(see \cref{fig:lakes})
instance generation is governed by three parameters:
the granularity~$M$ of the ``terrain'',
the triangular grid size~$N$,
and the number of obligatory components~$c$.
We generate a
Perlin noise function~$F\colon [0,M-1]^2 \mapsto \mathbb{R}$~\citep{PH89}
over a square with edge length~$M-1$.
Intuitively,
$F(x,y)$
gives the height of the terrain over point~$(x,y)\in[0,M-1]^2$
of the square (\cref{fig:lakesa}).
When generating the Perlin noise function~$F$,
we choose the angle of the gradient~$\nabla F(x,y)$
randomly for the \emph{inner integer} points~$(x,y)$ of the square.
In order to prevent lakes
from crossing the boundaries of the square,
we choose~$\nabla F(x,y)$
to point %
towards the square center for points~$(x,y)$ on the square's corners
and perpendicular to the
edge for the other integer points~$(x,y)$ on the square's edges,
as shown in \cref{fig:lakesa}.

\looseness=-1
We then cover the terrain using a triangular sensor grid
as in the ``faulty grid'' data set,
yet take only those grid nodes~$(x,y)\in [0,N]^2$
lying in lakes, that is,
$F(x\frac{M-1}{N}, y\frac{M-1}{N}) < 0$ (\cref{fig:lakesb} and \cref{fig:lakesc}).
Like in the ``faulty grid''
data set,
all generated network nodes form a vertex set~$V$
of a complete graph~$G$,
where the weight of each edge is the squared
Euclidean distance between the end points.
The described procedure leaves an instance
with a small number of obligatory components~$c$.
We repeat the generation process until the desired value of~$c$ is obtained.

We generated instances for
$M=7$,
$N\in\{10, 12, \dots, 30\}$, and
$c\in\{3, 4, 5\}$.

\subsection{Experimental setup.}
\label{sec:setup}
We implemented \cref{alg:ccs} in
approximately 650 lines of C++,
not including the code for the testing environment.%
\footnote{The code and testing environment are available at \url{https://gitlab.com/rvb/mpsc}.}
We compiled the code
using GNU C++ 9.2.1
with optimization level \texttt{-O2}.
The experiments were conducted
on a dual core Intel Core i7-9700K CPU
with 3.60\,GHz
and 15.6\,GiB of RAM running 64-bit Ubuntu~19.10.
We compared the following five algorithms.

\begin{description}
\item[DP1:] \cref{alg:ccs}
\emph{with} the the two data reduction rules described in \cref{sec:heur}.
\item[DP2:] \cref{alg:ccs}
\emph{without} the the two data reduction rules described in \cref{sec:heur}.
\item[BF:] The simple $O(n^{2c})$-time brute force algorithm described in the beginning of \cref{sec:outl}.
\item[EX1, EX2:]  Two state-of-the-art exact
  algorithms based on ILP models
  suggested by \citet{MG05},
  which we solve using CPLEX 12.8.\footnote{\url{http://www.cplex.com}}
\end{description}

\noindent
For DP1 and DP2,
we chose
$\varepsilon=0.1$ as an upper bound on the error probability.
We implemented EX1 and EX2 of \citet{MG05}
with the edge reduction rule from \cref{sec:heur}
and their extra inequalities (18), (19), (20), (23), (24), (25),
as they were the most effective
in their experiments.
We point out that our implementations of algorithms EX1 and EX2
also make use of vertex lower bounds:
we add ILP constraints
that force edges of the obligatory subgraph
into the solution.

In particular,
this means that
our implementation of algorithm EX2,
just like our algorithms DP1 and DP2,
merely has to solve the problem of connecting $c$~components:
EX2 originally consists of iteratively solving an ILP
model, adding more and more connectivity constraints in each iteration.
The solution to the ILP model in each iteration
yields a spanning subgraph that is not necessarily connected.
Its connected components are computed
and, for the next iteration,
a constraint is added so that at least one edge
has to leave each connected component.
Since our implementation in the first iteration adds 
the constraints that force
all obligatory edges into the solution,
our implementation of
EX2 has all obligatory connected components available
after the first iteration.

We ran the algorithms on instances
described in \cref{sec:data}.
We generated 10 instances with different random seeds
for each set of generation parameters.

\subsection{Experimental results.}
\label{sec:results}

\begin{figure}
  \centering
  \small
  \ref{algos}
  
  \begin{subfigure}{0.32\textwidth}
  \begin{tikzpicture}[scale=1]
    \begin{axis}[xlabel={Number of vertices},
                ylabel={Average running time [s]}, %
                width=\textwidth,
                height=\textwidth,
                ymin=-1, ymax=12.5,
                legend columns=-1,
                legend to name=algos]
      \addplot[mark=x] table [x=dp1_vertices_before_mean, y expr=\thisrow{dp1_transform_time_mean}+\thisrow{dp1_solve_time_mean}] {simple_C_3_dp1.txt}; \addlegendentry[scale=0.8]{DP1}
      \addplot[mark=+] table [x=dp2_vertices_before_mean, y expr=\thisrow{dp2_transform_time_mean}+\thisrow{dp2_solve_time_mean}] {simple_C_3_dp2.txt}; \addlegendentry[scale=0.8]{DP2}
      \addplot[mark=square] table [x=ex2lb_vertices_before_mean, y expr=\thisrow{ex2lb_transform_time_mean}+\thisrow{ex2lb_solve_time_mean}] {simple_C_3_ex2lb.txt}; \addlegendentry[scale=0.8]{EX2}
      \addplot[mark=o] table [x=ex1lb_vertices_before_mean, y expr=\thisrow{ex1lb_transform_time_mean}+\thisrow{ex1lb_solve_time_mean}] {simple_C_3_ex1lb.txt}; \addlegendentry[scale=0.8]{EX1}
      \addplot[mark=star] table [x=bf_vertices_before_mean, y expr=\thisrow{bf_transform_time_mean}+\thisrow{bf_solve_time_mean}] {simple_C_3_bf.txt}; \addlegendentry[scale=0.8]{BF}
    \end{axis}
  \end{tikzpicture}
  \caption{$c=3$}
  \end{subfigure}
  \hfill
  \begin{subfigure}{0.32\textwidth}
  \begin{tikzpicture}[scale=1]
    \begin{axis}[xlabel={Number of vertices}, ylabel={Average running time [s]},
                 width=\textwidth,
                 height=\textwidth,
                 ymin=-1, ymax=12.5]
      \addplot[mark=x] table [x=dp1_vertices_before_mean, y expr=\thisrow{dp1_transform_time_mean}+\thisrow{dp1_solve_time_mean}] {simple_C_4_dp1.txt}; %
      \addplot[mark=+] table [x=dp2_vertices_before_mean, y expr=\thisrow{dp2_transform_time_mean}+\thisrow{dp2_solve_time_mean}] {simple_C_4_dp2.txt};%
      \addplot[mark=square] table [x=ex2lb_vertices_before_mean, y expr=\thisrow{ex2lb_transform_time_mean}+\thisrow{ex2lb_solve_time_mean}] {simple_C_4_ex2lb.txt}; %
      \addplot[mark=o] table [x=ex1lb_vertices_before_mean, y expr=\thisrow{ex1lb_transform_time_mean}+\thisrow{ex1lb_solve_time_mean}] {simple_C_4_ex1lb.txt}; %
    \end{axis}
  \end{tikzpicture}
  \caption{$c=4$}
  \end{subfigure}
\hfill
  \begin{subfigure}{0.32\textwidth}
  \begin{tikzpicture}[scale=1]
    \begin{axis}[xlabel={Number of vertices}, ylabel={Average running time [s]},
                 width=\textwidth,
                 height=\textwidth,
                 ymin=-1, ymax=12.5]
      \addplot[mark=x] table [x=dp1_vertices_before_mean, y expr=\thisrow{dp1_transform_time_mean}+\thisrow{dp1_solve_time_mean}] {simple_C_5_dp1.txt}; %
      \addplot[mark=square] table [x=ex2lb_vertices_before_mean, y expr=\thisrow{ex2lb_transform_time_mean}+\thisrow{ex2lb_solve_time_mean}] {simple_C_5_ex2lb.txt}; %
      \addplot[mark=o] table [x=ex1lb_vertices_before_mean, y expr=\thisrow{ex1lb_transform_time_mean}+\thisrow{ex1lb_solve_time_mean}] {simple_C_5_ex1lb.txt}; %
    \end{axis}
  \end{tikzpicture}
  \caption{$c=5$}
  \end{subfigure}
  \hfill
  \caption{Experimental evaluation of the running times of the five
    implementations BF, DP1, DP2, EX1, and EX2
    on ``faulty grid'' instances with \(C\in\{3,4,5\}\) connected components
    in the obligatory subgraph.
    }
  \label{fig:results}
\end{figure}

\subsubsection{``Faulty grid'' instances.}
\label{sec:faultyres}
\cref{fig:results} shows our experimental results
on the ``faulty grid'' instances.
The brute force algorithm~BF
is among the slowest already for \(c=3\)
and is excluded from the plots for $c\geq 4$
since its running time on only $332$~vertices
varied from $45$~seconds to two hours.
DP2, although being one of the best algorithms for \(c=3\),
is a bad choice already for \(c\geq4\),
and is excluded from the plot for~$c=5$,
since its running time for $2221$ vertices
varied between $212$ seconds and $85$ minutes.
DP1 outperforms
the integer linear programming models EX1 and EX2
of \citet{MG05}
for all~\(c\in\{3,4,5\}\).
This supports the claim that our
fixed\hyp parameter algorithm efficiently
solves instances for small values of~\(c\).

\begin{figure}
  \centering
  \small
	\def\maxValueRun{3000}
	\def\minValueRun{0.00003}
	\def\smallFactor{6}
    \def\bigFactor{36}

    \begin{subfigure}{0.48\textwidth}
	\begin{tikzpicture}[scale=1]
          \begin{loglogaxis}[
            width=\textwidth,
            height=\textwidth,
            xlabel={DP1 running time [s]},
            ylabel={EX2 running time [s]},
            legend cell align=left,
            legend pos=north west,
            xmax=30,%
            xmin=\minValueRun,
            ymax=30,%
            ymin=\minValueRun,
            ymode=log
            ]
            \addplot[only marks,mark=+, mark options={black!50}]
                table[col sep=tab,
                      x expr=\thisrow{dp1_solve_time}+\thisrow{dp1_transform_time},
                      y expr=\thisrow{ex2lb_solve_time}+\thisrow{ex2lb_transform_time}]
                      {perlin_C_3_dp1_ex2lb.txt};
            \addplot[color=black,domain=\minValueRun:\maxValueRun,samples=4] {x};
            \addplot[dashed,color=black!75,domain=\minValueRun:\maxValueRun,samples=4] {\smallFactor*x};
            \addplot[dashed,color=black!75,domain=\minValueRun:\maxValueRun,samples=4] {x/\smallFactor};
            \addplot[dotted,color=black,domain=\minValueRun:\maxValueRun,samples=4] {\bigFactor*x};
            \addplot[dotted,color=black,domain=\minValueRun:\maxValueRun,samples=4] {x/\bigFactor};
          \end{loglogaxis}
	\end{tikzpicture}
    \caption{$c = 3$}
    \end{subfigure}\hfill
    \begin{subfigure}{0.48\textwidth}
	\begin{tikzpicture}[scale=1]
          \begin{loglogaxis}[
            width=\textwidth,
            height=\textwidth,
            xlabel={DP1 running time [s]},
            ylabel={EX2 running time [s]},
            legend cell align=left,
            legend pos=north west,
            xmax=30,%
            xmin=0.0003,%
            ymax=30,%
            ymin=0.0003,%
            ymode=log
            ]
            \addplot[only marks,mark=x,mark options={black!75}]
                table[col sep=tab,
                      x expr=\thisrow{dp1_solve_time}+\thisrow{dp1_transform_time},
                      y expr=\thisrow{ex2lb_solve_time}+\thisrow{ex2lb_transform_time}]
                      {perlin_C_4_dp1_ex2lb.txt};
            \addplot[color=black,domain=\minValueRun:\maxValueRun,samples=4] {x};
            \addplot[dashed,color=black!75,domain=\minValueRun:\maxValueRun,samples=4] {\smallFactor*x};
            \addplot[dashed,color=black!75,domain=\minValueRun:\maxValueRun,samples=4] {x/\smallFactor};
            \addplot[dotted,color=black,domain=\minValueRun:\maxValueRun,samples=4] {\bigFactor*x};
            \addplot[dotted,color=black,domain=\minValueRun:\maxValueRun,samples=4] {x/\bigFactor};
          \end{loglogaxis}
	\end{tikzpicture}
    \caption{$c = 4$}
    \end{subfigure}\hfill

    \bigskip
    \noindent
    \begin{subfigure}{0.48\textwidth}
	\begin{tikzpicture}[scale=1]
          \begin{loglogaxis}[
            width=\textwidth,
            height=\textwidth,
            xlabel={DP1 running time [s]},
            ylabel={EX2 running time [s]},
            legend cell align=left,
            legend pos=north west,
            xmax=\maxValueRun,
            xmin=0.003,%
            ymax=\maxValueRun,
            ymin=0.003,%
            ymode=log
            ]
            \addplot[only marks,mark=star, mark options={black}]
                table[col sep=tab,
                      x expr=\thisrow{dp1_solve_time}+\thisrow{dp1_transform_time},
                      y expr=\thisrow{ex2lb_solve_time}+\thisrow{ex2lb_transform_time}]
                      {perlin_C_5_dp1_ex2lb.txt};
            \addplot[color=black,domain=\minValueRun:\maxValueRun,samples=4] {x};
            \addplot[dashed,color=black!75,domain=\minValueRun:\maxValueRun,samples=4] {\smallFactor*x};
            \addplot[dashed,color=black!75,domain=\minValueRun:\maxValueRun,samples=4] {x/\smallFactor};
            \addplot[dotted,color=black,domain=\minValueRun:\maxValueRun,samples=4] {\bigFactor*x};
            \addplot[dotted,color=black,domain=\minValueRun:\maxValueRun,samples=4] {x/\bigFactor};
          \end{loglogaxis}
	\end{tikzpicture}
    \caption{$c = 5$}
    \end{subfigure}\hfill
    \begin{subfigure}{0.48\textwidth}
	\begin{tikzpicture}[scale=1]
          \begin{loglogaxis}[
            width=\textwidth,
            height=\textwidth,
            xlabel={DP1 running time [s]},
            ylabel={EX2 running time [s]},
            legend cell align=left,
            legend pos=north west,
            xmax=\maxValueRun,
            xmin=\minValueRun,
            ymax=\maxValueRun,
            ymin=\minValueRun,
            ymode=log
            ]
            \addplot[only marks,mark=+,mark options={black!50}]
                table[col sep=tab,
                      x expr=\thisrow{dp1_solve_time}+\thisrow{dp1_transform_time},
                      y expr=\thisrow{ex2lb_solve_time}+\thisrow{ex2lb_transform_time}]
                      {perlin_C_3_dp1_ex2lb.txt};
            \addplot[only marks,mark=x,mark options={black!75}]
                table[col sep=tab,
                      x expr=\thisrow{dp1_solve_time}+\thisrow{dp1_transform_time},
                      y expr=\thisrow{ex2lb_solve_time}+\thisrow{ex2lb_transform_time}]
                      {perlin_C_4_dp1_ex2lb.txt};
            \addplot[only marks,mark=star]
                table[col sep=tab,
                      x expr=\thisrow{dp1_solve_time}+\thisrow{dp1_transform_time},
                      y expr=\thisrow{ex2lb_solve_time}+\thisrow{ex2lb_transform_time}]
                      {perlin_C_5_dp1_ex2lb.txt};
            \addplot[color=black,domain=\minValueRun:\maxValueRun,samples=4] {x};
            \addplot[dashed,color=black!75,domain=\minValueRun:\maxValueRun,samples=4] {\smallFactor*x};
            \addplot[dashed,color=black!75,domain=\minValueRun:\maxValueRun,samples=4] {x/\smallFactor};
            \addplot[dotted,color=black,domain=\minValueRun:\maxValueRun,samples=4] {\bigFactor*x};
            \addplot[dotted,color=black,domain=\minValueRun:\maxValueRun,samples=4] {x/\bigFactor};
          \end{loglogaxis}
	\end{tikzpicture}
        \caption{All instances for $c\in\{3, 4, 5\}$ on one plot.}
    \end{subfigure}\hfill

    \caption{Comparison of EX2 and DP1 on the ``lakes'' instances.
      Each point is a single problem instance.
      DP1 is faster than EX2 if the point is above the diagonal
      and slower otherwise.
      The dashed line indicates running time
      difference by a factor of~$\smallFactor$,
      the dotted line shows a factor of~$\bigFactor$.}
  \label{fig:comparison}
\end{figure}

\subsubsection{``Lakes'' instances.}
\label{sec:lakeres}
Since DP1 and EX2 were the fastest algorithms
in \cref{sec:faultyres},
we compare them in more detail on the ``lakes'' data set.
\cref{fig:comparison} shows the running time
of DP1 and EX2
with $c\in\{3, 4, 5\}$.
DP1 is at least 18.75~times faster
than
EX2 %
on instances with $c=3$ and
at least~$3.75$ times faster on instances with $c=4$.
For $c=3$, most instances are solved at least $36$~times faster by DP1 than by EX2. %
For $c=4$, most instances %
are solved from $6$ to $36$~times faster by DP1.
For $c=5$, EX2 is faster on small instances,
yet in all cases,
we see that DP1 is faster on the instances
that are hard to solve by both algorithms.

\begin{figure}[t]
  \centering
  \small
  \input{timespread-001.tex}
  \hfill
  \input{timespread-002.tex}
  \caption{Running times of DP1 and DP2.
    Each point is an instance of ``lakes'' data set with $c=5$.
    Estimates of the mean and the standard deviation are shown.}
    \label{fig:boxes}
\end{figure}

\cref{fig:boxes} shows
the running time of DP1 and EX2
in dependence on the number of input vertices,
estimates of the mean and the standard deviation.
We see not only that the running time of EX2
clearly grows superpolynomially with the number of vertices,
but also has a higher variance than that of DP1.
The running time of EX2 can easily vary by a factor of 100
for graphs of the same size, whereas that of DP1 varies by a factor of 10.
This makes the running times of DP1 more predictable.

\begin{figure}%
  \centering
  \input{timespread-003.tex}
  \caption{Effect of the data reduction described in \cref{sec:heur} on the ``lakes'' instances with~$c=5$.  Estimates of mean and standard deviation are shown.}
  \label{fig:reduction_effect}
\end{figure}

\subsubsection{The role of data reduction.}
\label{sec:dr-effect}
One reason for the better performance of DP1 compared to EX2
surely is the data reduction  presented in \cref{sec:heur}.
On \cref{fig:reduction_effect},
we see that the heavy edge deletion rule
reduces the number of edges to about 25\,\%.
Additionally removing redundant vertices
reduces the instance size to about 5\,\%.
Yet the vertex deletion rule
is applicable only to DP1.
Applying it to EX1 or EX2
would break the connectivity constraints %
and require
cardinally changing the models
and the additional inequalities for speeding them up.
This is beyond the scope of our work,
which is focused on algorithms with
provable running time bounds.
Moreover, data reduction is not the only reason
for the better performance of DP1:
\cref{fig:boxes} clearly exhibits
a superpolynomial running time dependency of EX2 on the number~$n$ of vertices,
whereas DP1 has a proven \emph{worst-case} running time $O(nm\log n)$
for fixed~$c$.

\subsubsection{Error rate.}
We ran the randomized algorithms DP1 and DP2
with an upper bound of~$\varepsilon=0.1$ on the error probability,
yet in fact the empirical error rate was significantly lower.
Neither DP1 nor DP2
gave incorrect answers on any of the ``faulty grid''
instances.

DP1 yielded incorrect answers on 3.6\,\% of 330
samples from the ``lakes'' data set,
whereas DP2 yielded incorrect answers on 2.5\,\% of 320 samples
(the number of conducted experiments for DP2 is lower
since it was unable to finish some large instances
with \(c=5\) in reasonable time).
The cost of incorrect solutions was higher than
the cost of an optimal solution by at most 3.5\,\%.
We point out that,
by merely doubling the running time of DP1,
one can guarantee an error rate below~$\varepsilon=0.01$
and still significantly outperform EX2 on large instances.

\subsubsection{Conclusion.}
For small~$c$,
DP1 obviously outperforms EX2.
For larger~$c$,
the advantage of DP1 over EX2 becomes smaller.
Yet for large enough instances,
DP1 will always
outperform EX2:
the running time of DP1
for constant~$c$ is merely $O(nm\log n)$,
whereas for EX2,
one has to expect the running time
to depend superpolynomially on~$n$,
as witnessed by \cref{fig:boxes}.

In general,
we thus recommend to use DP1 for exactly solving \mpsc{} instances
with obligatory subgraphs with a small number of connected components
and for large graphs.

\section{Conclusion}
We presented a new algorithm for \miPoSyCo{} that
runs in polynomial time
on instances in which we can find an obligatory subgraph
with logarithmically many connected components.
On instances with few such connected components,
it outperforms state-of-the art integer linear programming models.
To achieve this,
data reduction played a crucial role,
yet we also saw that data reduction with \emph{provable}
effect is hard.

Our algorithms are less suited 
for random test data (as typically used in published work so far) 
because our algorithms make explicit use of 
structure in the input that plausibly occurs in
real-world monitoring instances,
where the layout of the sensors in the network
has to take into account energy-efficiency in order to
maximize the lifetime of the sensor network.

An important theoretical and practical
challenge is to find \plb{}s
that yield obligatory subgraphs with few connected components.
This goes hand in hand with 
identifying scenarios where (more) obligatory 
edges are given by the application.
We identified the scenario
where a sensor network lost connectivity
and has to be reconnected
at minimum additional energy consumption,
but this also may be the case in communication 
networks with designated hub nodes. 

\paragraph{Acknowledgments.}
  The results in
  \cref{sec:dr,sec:heur,sec:dr-effect} were obtained while
  René van Bevern
  and Pavel V.\ Smirnov
  were supported by the
  Russian Foundation for Basic Research
  under grant
  18-501-12031 NNIO\textunderscore a.
  The results in
  \cref{sec:obl,sec:outl,sec:corr,sec:dp,sec:faulty,sec:faultyres}
  were obtained
  while
  René van Bevern
  was supported by
  stipend SP-2178.2019.5
  of the
  President of the Russian Federation
  and Pavel V.\ Smirnov
  was supported by Mathematical Center in Akademgorodok,
  agreement No.\ 075-15-2019-1675 with the Ministry of Science and
  Higher Education of the Russian Federation. 
  The results in \cref{sec:hardness,sec:lakes,sec:lakeres}
  were obtained while both were supported by Mathematical Center in Akademgorodok.
  Both
  thank Oxana Yu.\ Tsidulko for fruitful discussions.

\bibliographystyle{informs}
\bibliography{mylib}

\end{document}